\documentclass[11pt]{article}

\usepackage[T1]{fontenc}
\usepackage{lmodern}
\usepackage[margin=1in]{geometry}
\usepackage{amsmath,amssymb,amsthm,mathtools,bm}
\usepackage{aliascnt}
\usepackage{algorithm}
\usepackage{algpseudocode}
\usepackage{microtype}
\usepackage{enumitem}
\usepackage{xcolor}
\usepackage{hyperref}
\usepackage[nameinlink,capitalise,noabbrev]{cleveref}

\definecolor{newblue}{rgb}{0.19,0.55,0.91}

\hypersetup{
    colorlinks=true,
    linkcolor=blue,
    citecolor=newblue,
    urlcolor=blue!60!black
}

\numberwithin{equation}{section}

\newtheorem{theorem}{Theorem}[section]
\newaliascnt{lemma}{theorem}
\newtheorem{lemma}[lemma]{Lemma}
\aliascntresetthe{lemma}
\newaliascnt{proposition}{theorem}
\newtheorem{proposition}[proposition]{Proposition}
\aliascntresetthe{proposition}
\newaliascnt{corollary}{theorem}
\newtheorem{corollary}[corollary]{Corollary}
\aliascntresetthe{corollary}
\theoremstyle{definition}
\newaliascnt{definition}{theorem}
\newtheorem{definition}[definition]{Definition}
\aliascntresetthe{definition}
\newaliascnt{remark}{theorem}
\newtheorem{remark}[remark]{Remark}
\aliascntresetthe{remark}

\crefname{theorem}{Theorem}{Theorems}
\Crefname{theorem}{Theorem}{Theorems}
\crefname{lemma}{Lemma}{Lemmas}
\Crefname{lemma}{Lemma}{Lemmas}
\crefname{proposition}{Proposition}{Propositions}
\Crefname{proposition}{Proposition}{Propositions}
\crefname{corollary}{Corollary}{Corollaries}
\Crefname{corollary}{Corollary}{Corollaries}
\crefname{definition}{Definition}{Definitions}
\Crefname{definition}{Definition}{Definitions}
\crefname{remark}{Remark}{Remarks}
\Crefname{remark}{Remark}{Remarks}
\crefname{section}{Section}{Sections}
\Crefname{section}{Section}{Sections}
\crefname{equation}{Equation}{Equations}
\Crefname{equation}{Equation}{Equations}

\DeclarePairedDelimiter\norm{\lVert}{\rVert}

\newcommand{\R}{\mathbb{R}}
\newcommand{\E}{\mathbb{E}}
\newcommand{\Prob}{\mathbb{P}}
\newcommand{\Prb}{\mathbb{P}}
\newcommand{\eps}{\varepsilon}
\newcommand{\A}{A}
\newcommand{\B}{B}
\newcommand{\Smat}{S}
\newcommand{\I}{I}
\newcommand{\Pmat}{P}
\newcommand{\arow}{a}
\newcommand{\evec}{e}
\newcommand{\OPT}{\mathsf{OPT}}
\newcommand{\cost}{\operatorname{cost}}
\newcommand{\Fset}{\mathcal{F}}
\newcommand{\nnz}{\operatorname{nnz}}
\newcommand{\wtO}{\widetilde{O}}
\newcommand{\wtOmega}{\widetilde{\Omega}}
\newcommand{\wtTheta}{\widetilde{\Theta}}

\newcommand{\cH}{\mathcal{H}}
\newcommand{\cP}{\mathcal{P}}
\newcommand{\Gr}{\operatorname{Gr}}
\newcommand{\dist}{\operatorname{dist}}
\newcommand{\rank}{\operatorname{rank}}

\newcommand{\diag}{\operatorname{diag}}
\newcommand{\Rad}{\mathfrak{R}}
\newcommand{\Pdim}{\operatorname{Pdim}}
\newcommand{\sgn}{\operatorname{sgn}}
\newcommand{\polylog}{\operatorname{polylog}}
\newcommand{\ip}[2]{\left\langle #1,#2\right\rangle}
\newcommand{\abs}[1]{\left\lvert #1\right\rvert}
\newcommand{\1}{\mathbf{1}}

\title{Nearly Optimal Strong Coresets for $\ell_p$ Subspace Approximation}
\author{
Honghao Lin\footnote{Google Research, Carnegie Mellon University / Texas A\&M University}
\and
Vahab Mirrokni\footnote{Google Research}
\and
David P. Woodruff\footnote{Google Research and Carnegie Mellon University}
}
\date{}

\begin{document}
\maketitle

\begin{abstract}
We study strong coresets for $\ell_p$ subspace approximation. Given
$\A\in\R^{n\times d}$, the goal is to sample and rescale a small number of
its rows to form $\Smat\A$ such that
\[
    \norm{\Smat\A(\I-\Pmat_F)}_{p,2}^p
    =(1\pm\eps)\norm{\A(\I-\Pmat_F)}_{p,2}^p
\]
simultaneously for every subspace $F\subseteq\R^d$ of dimension at most $k$,
where $\Pmat_F$ is the orthogonal projector onto $F$.
Woodruff and Yasuda (FOCS 2025)~\cite{WY25} obtained coreset sizes
$\wtO_p(k\eps^{-4/p})$ for $1\leq p<2$ and
$\wtO_p(k^{p/2}\eps^{-p})$ for $p>2$.
We improve these bounds to $\wtO_p(k\eps^{-2})$ and
$\wtO_p(k^{p/2}\eps^{-2})$, respectively.
For $1\leq p<2$, our algorithm runs in
$\wtO_p(\nnz(\A)+d^\omega+k\eps^{-2})$ time. The resulting coreset size
matches the known sampling lower bound~\cite{LWW21} up to logarithmic factors
when
$k+1\geq C\log(1/\eps)$ for an absolute constant $C$.
For $p>2$, our algorithm runs in $\wtO_p(\nnz(\A)+d^\omega)$ time, matching
the running time of the Woodruff--Yasuda framework.

We use different techniques in the two regimes. For
$1\leq p<2$, we combine a bicriteria low-rank split with Lewis-weight sampling
and empirical-process bounds independent of the output dimension. For $p>2$,
we give a sharper analysis of the Woodruff--Yasuda construction. By retaining the
truncation in its sampling probabilities throughout the row-count recurrence,
we show that it achieves the improved $\eps^{-2}$ dependence.
\end{abstract}

\section{Introduction}

Subspace approximation is a basic primitive for representing high-dimensional
data by a low-dimensional linear model~\cite{DV07,DTV11,SV12,CW15}. Given a
collection of points, the goal is to find a low-dimensional subspace that
minimizes the aggregate distance of the points from that subspace. This model
underlies classical dimensionality reduction as well as robust variants of
low-rank approximation.

Let $\A\in\R^{n\times d}$ have rows $\arow_1^\top,\dots,\arow_n^\top$. For a
subspace $F\subseteq\R^d$, let $\Pmat_F$ denote the orthogonal projector onto
$F$. For $p\geq1$, the $\ell_p$ subspace approximation cost of $F$ is
\[
    \cost_{\A}(F)
    :=\norm{\A(\I-\Pmat_F)}_{p,2}^p
    =\sum_{i=1}^n
      \norm{\arow_i^\top(\I-\Pmat_F)}_2^p.
\]
Writing $\Fset_k$ for the set of subspaces of dimension at most $k$, the
optimization problem is
\[
    \OPT_A:=\min_{F\in\Fset_k}\cost_A(F).
\]
The case $p=2$ is principal component analysis and is solved by the singular
value decomposition. The case $p=1$ includes the median hyperplane problem,
while larger values of $p$ increasingly emphasize points with large residuals
and give finite-$p$ relaxations of the center hyperplane objective. Except for
the Euclidean case $p=2$, these problems are generally computationally
hard~\cite{DTV11,GRSW12,CW15}.

This hardness makes data reduction especially consequential. A common route
to an approximation algorithm is to first replace the input by a much smaller
instance and then run a more expensive optimization routine on that instance;
the size of the reduced instance therefore directly affects both the storage
cost and the eventual running time~\cite{FMSW10,CW15}. The central target is a
\emph{dimension-independent} summary: its size may depend on the intrinsic
rank $k$, the accuracy $\eps$, and the fixed exponent $p$, but not on either
the number of points $n$ or the ambient dimension $d$.

Coresets provide such summaries while retaining the geometric objective. A
row-sampling and rescaling matrix $\Smat$ represents a $(1\pm\eps)$ strong row
coreset if it is supported on few input rows and satisfies
\[
    \cost_{\Smat\A}(F)=(1\pm\eps)\cost_{\A}(F)
    \qquad\text{simultaneously for every }F\in\Fset_k.
\]
Equivalently, the coreset is a small weighted subset of the input points that
preserves every rank-at-most-$k$ subspace cost. If it contains $m$ rows, its
row span has dimension at most $m$, so rotational invariance also allows the
retained instance to be represented in at most $m$ dimensions. Thus a bound
$m=\operatorname{poly}(k/\eps)$ removes both $n$ and $d$ from the size of the optimization
instance.

The uniformity over $F$ is the defining strength of this guarantee. An
optimizer for the sampled instance immediately yields an approximate optimizer
for the original data, but the same sample can also be reused for many queries,
for an adaptively chosen candidate subspace, or for a restricted family of
subspaces imposed only after preprocessing. These applications are not covered
by a weak coreset, which need only preserve the optimum value or an approximate
minimizer~\cite{FL11,HV20}. A strong coreset is therefore a reusable surrogate
for the entire objective rather than only for one optimization run.

There is also an important distinction between a coreset supported on the
original input rows and a reduced instance consisting of modified points.
Earlier dimension-independent constructions append an extra coordinate to
each retained point~\cite{SW18,FKW21}. Such a construction still gives a
compact representation, but it is not a weighted subset of the original
dataset. A true row-subset coreset preserves the ambient representation and
the sparsity pattern of every retained row, maintains a direct correspondence
between sampled and original points, and can be passed to downstream routines
without adapting them to an additional coordinate. The distinction is also
theoretically meaningful: the known sampling lower bounds concern weighted
subsets of the input rows~\cite{LWW21,WY23}.

Obtaining dimension independence, a true row subset, nearly optimal rank
dependence, and an efficient construction simultaneously was the main obstacle.
Sohler and Woodruff~\cite{SW18} achieved dimension-independent size with nearly
optimal dependence on $k$, but used modified points and an exponential-time
construction. Subsequent work either retained the extra coordinate and lost
additional polynomial factors~\cite{FKW21}, or produced true row subsets with
polynomially worse dependence on $k$~\cite{HV20}. The natural benchmark is
$\wtO(k)$ rows for $p<2$ and $\wtO(k^{p/2})$ rows for $p>2$, with polynomial
dependence on $1/\eps$~\cite{LWW21,WY23}.

Woodruff and Yasuda~\cite{WY25} reached this benchmark in the rank parameter.
They obtained the first dimension-independent strong row coresets with nearly
optimal dependence on $k$ for every $p\ne2$, together with a
$\wtO(\nnz(\A)+d^\omega)$-time construction. Their framework recursively
samples rows according to root ridge leverage scores and returns
$\wtO_p(k\eps^{-4/p})$ rows for $1\leq p<2$. For $p>2$, its row count is
\[
    \frac{k^{p/2}}{\eps^p}
    \left(\log\frac{k}{\eps\delta}\right)^{O(p^2)}.
\]
This result resolves the dimension and rank dependence, but it leaves the
accuracy dependence open. For $1\leq p<2$, the $\eps^{-4/p}$ upper bound is
separated from the sampling lower bound $\wtOmega_p(k\eps^{-2})$ by a
polynomial factor~\cite{LWW21}. For $p>2$, their one-round sampling theorem is
applied recursively, and the analysis of the resulting row-count recurrence
produces the factor $\eps^{-p}$. It is not immediate whether this factor is
required by the sampling rule itself.

\begin{quote}
\emph{What is the right dependence on $\eps$ for dimension-independent strong
row coresets for $\ell_p$ subspace approximation?}
\end{quote}

\subsection{Related work}

Coresets and sketches for high-dimensional geometric approximation have been
studied through dimensionality reduction, sensitivity sampling, and matrix
sampling techniques~\cite{FMSW10,FL11,VX12}.
Earlier work developed sampling-based dimension reduction and approximation
algorithms specifically for subspace objectives~\cite{DV07,DTV11,SV12}.
For squared Euclidean objectives, Feldman, Schmidt, and Sohler developed
dimension-independent reductions for PCA and projective clustering~\cite{FSS20}.
Sohler and Woodruff subsequently obtained the first strong coresets for
subspace approximation whose size is independent of both $n$ and $d$ for
constant $p$, using a representative subspace and points with an appended
coordinate rather than a weighted subset of the original rows~\cite{SW18}.
Their construction achieves the desired dependence on $k$ but runs in
exponential time, and the appended coordinate means that downstream
algorithms must be adapted to the modified instance. Feng, Kacham, and
Woodruff reduced the running time to polynomial for sum-of-distances
objectives, while retaining the appended coordinate and losing additional
polynomial factors in the coreset size~\cite{FKW21}.

A complementary line of work seeks coresets supported on the input rows.
Huang and Vishnoi obtained the first dimension-independent strong coresets of
this form through sensitivity sampling, but their size loses polynomial
factors in the rank parameter~\cite{HV20}.
Woodruff and Yasuda developed offline and online subset-selection algorithms
and online strong coresets for $\ell_p$ subspace approximation~\cite{WY23}.
Lewis-weight sampling is a central tool for $\ell_p$ subspace embeddings when
$p<2$~\cite{CP15}, and Munteanu and Omlor's $\ell_2$-augmentation framework
provides a related $\wtO_p(d\eps^{-2})$ sensitivity bound for subspace
embeddings~\cite{MO24}.

\subsection{Our results}

Our results obtain $\eps^{-2}$ dependence in both regimes. For
$1\leq p<2$, the first theorem closes the gap with the sampling lower bound,
up to logarithmic factors. Its pre-sparsified implementation makes these
factors independent of $n$.

\begin{theorem}[Strong row coresets for $1\leq p<2$]
\label{thm:main-small-p}
Fix $1\leq p<2$.
Given $\A\in\R^{n\times d}$, $1\leq k<d$, and
$\eps\in(0,1/2)$, there is a randomized algorithm running in
\[
    \wtO_p\!\left(\nnz(\A)+d^\omega+k\eps^{-2}\right)
\]
time that outputs a reweighted row subset $\Smat\A$ supported on at most
\[
 C_pk\eps^{-2}
 \left(\log\frac{C_pk}{\eps}\right)^{C_p}
\]
rows such that, with probability at least $2/3$,
\[
    \cost_{\Smat\A}(F)=(1\pm\eps)\cost_{\A}(F)
    \qquad\text{simultaneously for every }F\in\Fset_k.
\]
The sampled rows and their weights do not depend on $F$.
\end{theorem}

When $k+1\geq C\log(1/\eps)$, the sampling lower bound of
Li, Wang, and Woodruff~\cite{LWW21} implies
$\wtOmega_p(k\eps^{-2})$ rows. Thus \cref{thm:main-small-p} is optimal in
that parameter range. The upper bound itself holds without these additional
restrictions.

For $p>2$, the $\eps^{-p}$ dependence is not intrinsic to the
Woodruff--Yasuda construction.
We keep their one-round sampling probabilities and recursive framework, but a
sharper row-count analysis yields $\eps^{-2}$ dependence while preserving the
$k^{p/2}$ dependence and nearly input-sparsity running time.

\begin{theorem}[Strong row coresets for $p>2$]
\label{thm:main-large-p}
Fix $p>2$.
There is a constant $C_p\geq 1$ such that the following holds.
Given $\A\in\R^{n\times d}$, $1\leq k<d$, and
$\eps,\delta\in(0,1/2)$, there is a randomized algorithm running in
\[
    \wtO_p\bigl(\nnz(\A)+d^\omega\bigr)
\]
time that outputs a row-sampling matrix $\Smat$ such that, with
probability at least $1-\delta$,
\[
    \cost_{\Smat\A}(F)=(1\pm\eps)\cost_{\A}(F)
    \qquad\text{simultaneously for every }F\in\Fset_k,
\]
and
\[
    \nnz(\Smat)
    \leq
    \frac{C_p k^{p/2}}{\eps^2}
    \left(\log\frac{C_p k}{\eps\delta}\right)^{p+5}.
\]
\end{theorem}

The logarithmic exponent $p+5$ is a convenient explicit bound and is not
optimized.
Thus \cref{thm:main-large-p} improves the best known row-subset upper bound
from $\eps^{-p}$ to $\eps^{-2}$ and reduces the logarithmic exponent from
$O(p^2)$ to $O(p)$. Taken together, the two theorems give the
dimension-independent size bound
\[
    \wtO_p\!\left(k^{\max\{1,p/2\}}\eps^{-2}\right)
\]
for every fixed $p\in[1,\infty)\setminus\{2\}$, with the regime-specific
running times stated above.

\subsection{Technical overview}

\paragraph{The case $1\leq p<2$.}
To see the main difficulty, consider sampling indices
$I_1,\ldots,I_m$ independently from a distribution $q$ and estimating, for
each query subspace $F$,
\[
    \cost_{\A}(F)
    =\sum_i\norm{Q_Fa_i}_2^p
    \quad\text{by}\quad
    \widehat C_F
    :=\frac1m\sum_{s=1}^m
      \frac{\norm{Q_Fa_{I_s}}_2^p}{q_{I_s}},
    \qquad Q_F=I-P_F.
\]
For any fixed $F$, this estimator is unbiased. If $q_i$ is chosen according
to the largest relative contribution that row $i$ can make to a query, then
no summand is too large, and Bernstein's inequality gives concentration for
that particular $F$. However, a strong coreset requires uniform concentration
over all $F\in\Fset_k$. A standard $\eps$-net argument would cover $\Fset_k$
and take a union bound over the net, but the required covering
number depends on the ambient dimension $d$. We instead need a
dimension-independent uniform concentration bound governed only by $k$.

We first explain the argument assuming that we are given a subspace $F_0$,
spanned by $\wtO(k)$ rows of $A$, with
\[
    r:=\dim(F_0)=\wtO(k),
    \qquad
    \OPT_A:=\min_{F\in\Fset_k}\cost_A(F),
    \qquad
    \norm{\A(I-P_{F_0})}_{p,2}^p\leq\Gamma_0\OPT_A,
\]
where $\Gamma_0$ depends only on $p$. Thus $F_0$ has low dimension and its
projection residual is within a constant factor of the optimal rank-$k$ cost.
For each input row $a_i$, a sketched regression computes a fitted vector
$b_i\in F_0$, and we set $e_i=a_i-b_i$. Let $B$ and $E$ be the matrices with
rows $b_i$ and $e_i$, respectively. The regression guarantee gives
\[
    \A=B+E,
    \qquad \operatorname{rowspan}(B)=F_0,
    \qquad \rank(B)=r,
    \qquad R:=\norm{E}_{p,2}^p
    =\sum_i\norm{e_i}_2^p\leq\Gamma\OPT_A,
\]
where $\Gamma$ depends only on $p$. This split serves two different purposes.
The rows $b_i$ of $B$ lie in an $r$-dimensional space and can therefore be
controlled by Lewis weights. The residual $E$ may have high rank, but its
total mass is bounded by every query cost, since
$R\leq\Gamma\OPT_A\leq\Gamma\cost_{\A}(F)$.

Let $w_i$ be the $\ell_p$ Lewis weight of $b_i$, and, when $R>0$, let
$\rho_i=\norm{e_i}_2^p/R$ be the fraction of the residual mass in row $i$.
The vector-valued Lewis-weight sensitivity bound in
\cref{lem:lewis-p2-sensitivity} shows that,
for every $F$,
\[
    \norm{Q_Fb_i}_2^p
    \leq w_i\sum_j\norm{Q_Fb_j}_2^p.
\]
Moreover, the triangle inequality and the bound on $R$ imply
$\sum_j\norm{Q_Fb_j}_2^p=O_{p,\Gamma}(\cost_{\A}(F))$.
Since $Q_F$ is a contraction,
$\norm{Q_Fe_i}_2^p\leq\norm{e_i}_2^p
=\rho_iR\leq\Gamma\rho_i\cost_{\A}(F)$.
Thus $w_i$ and $\rho_i$ control the low-rank and residual contributions of
row $i$, respectively.

Set $D=k+r$, define the unnormalized balanced scores
\[
    s_i\asymp_{p,\Gamma}w_i+D\rho_i,
    \qquad q_i=\frac{s_i}{\sum_j s_j}.
\]
Set
\[
    m=\wtO_p(D\eps^{-2}),
\]
with a sufficiently large hidden constant, and draw $I_1,\ldots,I_m$
independently with replacement from $q$, rescaling each sampled row
$a_{I_s}$ by $(mq_{I_s})^{-1/p}$.
Since $\sum_iw_i=r$ and $\sum_i\rho_i=1$, the total score is $O(D)$, and
the resulting distribution satisfies
\[
    q_i\gtrsim_{p,\Gamma}\frac{w_i}{D},
    \qquad q_i\gtrsim_{p,\Gamma}\rho_i.
\]
In particular, every normalized row contribution is bounded by $O(D)$. This
already gives the correct sample size for a fixed $F$. The remaining challenge
is to prove concentration simultaneously over all $F\in\Fset_k$ without
introducing a dependence on the ambient dimension.

For this purpose, normalize each positive-cost query and write
\[
    x_F(i)=\frac{Q_Fb_i}{q_i^{1/p}\cost_{\A}(F)^{1/p}},
    \qquad
    y_F(i)=\frac{Q_Fe_i}{q_i^{1/p}\cost_{\A}(F)^{1/p}},
\]
so that
\[
    Z_F(i):=\norm{x_F(i)+y_F(i)}_2^p
    =\frac{\norm{Q_Fa_i}_2^p}{q_i\cost_{\A}(F)},
    \qquad \E_q Z_F=1.
\]
For the samples above, the relative estimator is exactly
\[
    \frac{\widehat C_F}{\cost_{\A}(F)}
    =\frac1m\sum_{s=1}^m Z_F(I_s).
\]
It is not enough to analyze the two summands separately: they need not be
orthogonal, and their parallel components may reinforce or cancel. We instead
describe the nonlinear loss $Z_F(i)$ by four scalar coordinates:
\begin{enumerate}
    \item the low-rank power $\norm{x_F(i)}_2^p$;
    \item the power in the component of $y_F(i)$ parallel to $x_F(i)$;
    \item their signed first-order interaction; and
    \item the power in the component of $y_F(i)$ orthogonal to $x_F(i)$.
\end{enumerate}
An exact geometric identity recovers $Z_F(i)$ from these four quantities
through a Lipschitz map, including at $p=1$, when $x_F(i)=0$, and under exact
cancellation.

The advantage of this decomposition is that each scalar class has
dimension-independent complexity. The low-rank process is controlled by the
Lewis geometry of the rank-$r$ matrix $B$. Directional and duality-map
estimates control the two parallel interaction terms, while a row-augmented
subspace entropy bound controls the orthogonal residual term. The
corresponding Rademacher bounds scale with $r$ or $k$, rather than with the
dimension of the output space $F^\perp$. Vector contraction then combines the
four estimates and yields
\[
    \E\sup_F\left|\frac1m\sum_{s=1}^m Z_F(I_s)-1\right|
    \leq
    \wtO_p\!\left(
      \sqrt{\frac{k+\rank(B)}{m}}
      +\frac{k+\rank(B)}{m}
    \right).
\]
Since $D=k+r=\wtO(k)$, the choice above uses
$m=\wtO_p(k\eps^{-2})$ samples. Taking the hidden constant sufficiently
large, the displayed bound and Markov's inequality give a
$(1\pm\eps)$ approximation simultaneously for every $F\in\Fset_k$ with
constant probability. Merging repeated samples then yields a strong row
coreset supported on at most $m$ input rows.

It remains to obtain such an $F_0$ and implement the scores efficiently. The
fast bicriteria construction of Woodruff and Yasuda~\cite[Lemma~2.3]{WY25}
provides, with high probability, a subspace spanned by $\wtO(k)$ input rows
that satisfies the assumed dimension and residual bounds. Once $F_0$ is
available, standard sparse embeddings, Johnson--Lindenstrauss sketches, and
Lewis-weight routines give constant-factor estimates of the residual masses
and Lewis weights; see \cref{sec:algorithm}.

\paragraph{The case $p>2$.}
Woodruff and Yasuda~\cite{WY25} provide a one-round sampling primitive for
sparsifying a weighted active matrix. Given a current matrix
$\B\in\R^{N\times d}$ with rows $\mathbf b_i^\top$, it sets
\[
    \lambda_{\B}:=\frac{\norm{\B-\B_k}_F^2}{k},
    \qquad
    \tau_i:=\mathbf b_i^\top
    (\B^\top\B+\lambda_{\B}\I)^\dagger\mathbf b_i,
\]
where $\B_k$ is a best rank-$k$ approximation to $\B$ in Frobenius norm and
$\dagger$ denotes the Moore--Penrose pseudoinverse.
These are ridge leverage scores at the standard rank-$k$ regularization
scale~\cite{CMM17}.
The primitive samples row $i$ independently with probability
\[
    q_i
    =\min\left\{1,
      \frac{N^{p/2-1}\tau_i^{p/2}}{\alpha}
      \right\},
\]
rescaling a retained row by $q_i^{-1/p}$.
Their one-round theorem shows that the sampling step simultaneously preserves
the cost of every rank-$k$ subspace, up to a small multiplicative error and an
additive term proportional to the optimum cost.
Since every subspace cost is at least the optimum, the additive term can be
absorbed into the relative-error budget after tuning its coefficient.
The construction in~\cite{WY25} applies this primitive recursively. We
retain the same sampling rule but sharpen the analysis of the number of rows
retained in each round. The per-round errors and failure probabilities
compose across the recursion.

The parameter $\alpha$ is the sampling threshold in the generic-chaining
argument that controls the deviation uniformly over all
$F\in\Fset_k$. For a one-round target error $\zeta$ and failure probability
$\eta$, it is chosen as
\[
    \alpha
    =\Theta_p\!\left(
      \frac{\zeta^2}{(\log N)^3+\log(1/\eta)}
    \right).
\]
Thus, up to logarithmic factors, $\alpha$ scales quadratically with the target
accuracy. The expected number of rows that survive one round is $\sum_i q_i$,
so the recursive coreset size is governed by how this thresholded sum is
bounded.
The available global information about the scores is the trace bound
$\sum_i\tau_i\leq 2k$.

For $0<\alpha\le1$, the earlier row-count bound in~\cite{WY25} uses the
relaxation
\[
    \sum_i q_i
    \leq
    \frac{1}{\alpha}
    \sum_i\min\{1,N^{p/2-1}\tau_i^{p/2}\}
    =O_p\!\left(\frac{k}{\alpha}N^{1-2/p}\right).
\]
Writing $N_t$ for the number of active rows after round $t$, iterating this
estimate gives a row-count recurrence of the form
\[
    N_t\lesssim k\alpha^{-1}N_{t-1}^{1-2/p},
\]
whose fixed point is
$k^{p/2}\alpha^{-p/2}$.
After the overall error $\eps$ is distributed across the recursive rounds,
$\alpha=\wtTheta_p(\eps^2)$.
The resulting fixed point therefore has an $\eps^{-p}$ dependence.

Our analysis keeps the full factor $N^{p/2-1}/\alpha$ inside the truncation.
At a high level, this charges both sampling regimes to the same linear score
budget. Rows whose sampling probabilities saturate at one must each carry a
substantial amount of ridge leverage mass, so the trace bound limits how many
such rows there can be. For the remaining rows, the superlinear contribution
$\tau_i^{p/2}$ can also be dominated by a suitably scaled linear contribution
in $\tau_i$. The following pointwise inequality captures both regimes at once.
For $p>2$ and every $y\geq 0$,
\[
    \min\{1,y^{p/2}\}\leq y.
\]
Taking $y=c^{2/p}\tau_i$ yields, for every $c>0$,
\[
    \min\{1,c\tau_i^{p/2}\}\leq c^{2/p}\tau_i.
\]
Applying this inequality directly with
$c=N^{p/2-1}/\alpha$ and then using the trace bound gives
\[
    \sum_iq_i
    \leq
    2k\alpha^{-2/p}N^{1-2/p},
\]
and the resulting recurrence is
\[
    N_t\lesssim k\alpha^{-2/p}N_{t-1}^{1-2/p}.
\]
Its fixed point is
\[
    \bigl(k\alpha^{-2/p}\bigr)^{p/2}
    =k^{p/2}\alpha^{-1},
\]
which gives $\eps^{-2}$ after substitution.
Thus the gain comes from preserving the interaction between truncation and
the chaining threshold before the recurrence is iterated: the threshold
contributes $\alpha^{-2/p}$ rather than $\alpha^{-1}$ to the row-count
recurrence, and the fixed-point calculation turns this local saving into the
improvement from $\eps^{-p}$ to $\eps^{-2}$.

\paragraph{Removing the $\log n$ dependence.}
Both analyses above contain logarithmic factors in the number $N$ of rows to
which the sampling procedure is applied. Running them directly on $A$ would
therefore leave a dependence on $\log n$ in the final coreset size. We first
apply the Woodruff--Yasuda strong coreset construction stated in
\cref{thm:WY-presparsification} with accuracy $\eps/3$, obtaining a
weighted row subset $A_0$ with $\operatorname{poly}_p(k/\eps)$ rows. We then
apply the corresponding construction described above to $A_0$. Since
$\log(\operatorname{rows}(A_0))=O_p(\log(k/\eps))$, all logarithmic factors in
the second stage depend only on $k$ and $\eps$.

\subsection{Open problems}

For $p>2$, the known lower bound for strong coresets is
\[
    \wtOmega\!\left(
      \frac{k^{p/2}}{\eps}+\frac{k}{\eps^2}
    \right),
\]
obtained from lower bounds for $\ell_p$ subspace embeddings~\cite{LWW21,WY23}.
This does not match our upper bound
$\wtO_p(k^{p/2}\eps^{-2})$ as a joint function of $k$ and $\eps$.
The main remaining question is whether the $p>2$ upper bound can be improved
to match the known lower bound up to logarithmic factors, or whether a
stronger lower bound captures an unavoidable interaction between the rank and
accuracy parameters. In contrast, for $1\leq p<2$ our upper bound matches the
known sampling lower bound under the parameter conditions stated above.

\paragraph{Organization.}
\Cref{sec:prelim} collects the sampling and sketching models, Lewis-weight
facts, and empirical-process tools used throughout the paper.
\Cref{sec:small-p-split,sec:algorithm} prove the upper bound for
$1\leq p<2$ and record the matching known lower bound.
The following three sections treat $p>2$: the one-round sampling section
states the baseline primitive from~\cite{WY25}, \cref{sec:threshold} proves
the sharper row-count estimate, and \cref{sec:construction} completes the
recursive construction. The appendices contain the empirical-process and
entropy estimates used in the $p<2$ analysis.

\section{Preliminaries}
\label{sec:prelim}

\subsection{Notation and sampling models}

Throughout the paper, $p$ is fixed within the regime under consideration.
Constants denoted by $c_p$ and $C_p$ may depend on $p$ but on no other
parameter, and their values may change from one occurrence to the next.
We write $\nnz(\A)$ for the number of nonzero entries of $\A$ and use $\omega$
for the exponent of matrix multiplication.
The notation $\wtO$, $\wtOmega$, and $\wtTheta$ suppresses logarithmic factors
in the relevant parameters. A subscript $p$ indicates that the implicit
constant may depend on $p$.
We write
\[
    \norm{Y}_{p,2}
    :=\left(\sum_i\norm{\evec_i^\top Y}_2^p\right)^{1/p}.
\]
For $F\subseteq\R^d$, write $P_F$ for the orthogonal projector onto $F$ and
$Q_F=I-P_F$. Let $\Gr(j,d)$ denote the set of $j$-dimensional subspaces of
$\R^d$. The set of subspaces of $\R^d$ of dimension at most $k$ is
\[
    \Fset_k:=\bigcup_{j=0}^k\Gr(j,d).
\]
We use the equivalent notation
\[
    \Phi_A(F)=\cost_A(F)=\norm{AQ_F}_{p,2}^p,
    \qquad
    \OPT_A:=\min_{F\in\Fset_k}\cost_A(F).
\]
We use the convention that
\[
    X=(1\pm a)Y\pm b
\]
means
$(1-a)Y-b\leq X\leq(1+a)Y+b$.
All logarithms are natural.

\begin{definition}[Strong row coreset]
A reweighted row subset is represented by a matrix
$\Smat\in\R^{s\times n}$ with one nonzero entry in each row. Repeated input
indices may be aggregated, and the size is the number of distinct retained
rows. It is a $(1\pm\eps)$ strong row coreset for rank-at-most-$k$
$\ell_p$ subspace approximation if
\[
    \cost_{\Smat\A}(F)
    =(1\pm\eps)\cost_{\A}(F)
    \qquad\text{for every }F\in\Fset_k.
\]
\end{definition}

\begin{definition}[Bernoulli $\ell_p$ sampling]
Given inclusion probabilities $\pi_1,\dots,\pi_N\in[0,1]$, let
$\xi_1,\dots,\xi_N$ be independent Bernoulli variables with
$\Prob[\xi_i=1]=\pi_i$.
The associated $\ell_p$ sampling matrix is diagonal with
\[
    \Smat_{ii}
    :=
    \begin{cases}
        \pi_i^{-1/p}\xi_i, & \pi_i>0,\\
        0, & \pi_i=0.
    \end{cases}
\]
After sampling, zero rows are deleted from the active matrix.
A product of successive sampling maps can be viewed, after padding deleted
rows with zeros, as a single diagonal reweighting of the original rows.
\end{definition}

\begin{definition}[Sampling with replacement]\label{def:iid-sampling}
Let $q$ be a probability distribution on $[N]$, and draw
$I_1,\ldots,I_m$ independently from $q$. The corresponding sampling matrix
$\Smat\in\R^{m\times N}$ has rows
\[
    \Smat_{s,*}=(mq_{I_s})^{-1/p}\evec_{I_s}^\top,
    \qquad s\in[m].
\]
Thus, if $M$ has rows $z_i^\top$, then
\[
    \norm{\Smat M}_{p,2}^p
    =\frac1m\sum_{s=1}^m\frac{\norm{z_{I_s}}_2^p}{q_{I_s}}.
\]
Repeated draws may be merged, so the support contains at most $m$ distinct
input rows.
\end{definition}

\subsection{Pre-sparsification and composition}

We use the following result of Woodruff and Yasuda as a preprocessing step.
The theorem is stated for weighted inputs by absorbing each positive row
weight into the corresponding row scaling.

\begin{theorem}[Woodruff--Yasuda pre-sparsification~\cite{WY25}]
\label{thm:WY-presparsification}
Fix $p\in[1,\infty)\setminus\{2\}$. Given
$A\in\R^{N\times d}$, $1\le k<d$, and
$\xi,\delta\in(0,1/2)$, there is a randomized algorithm that returns a
weighted row subset $A_0=S_0A$ satisfying, with probability at least
$1-\delta$,
\[
 \cost_{A_0}(F)=(1\pm\xi)\cost_A(F)
 \qquad\text{for every }F\in\Fset_k.
\]
For a constant $C_p$, its support size is at most
\begin{equation}\label{eq:WY-presparsification-size}
 \begin{cases}
  C_pk\xi^{-4/p}
  \left(\log\dfrac{C_pk}{\xi\delta}\right)^{C_p},
     &1\le p<2,\\[6pt]
  C_pk^{p/2}\xi^{-p}
  \left(\log\dfrac{C_pk}{\xi\delta}\right)^{C_p},
     &p>2.
 \end{cases}
\end{equation}
The running time is $\wtO_p(\nnz(A)+d^\omega)$.
\end{theorem}

This is the row-subset strong-coreset guarantee in
\cite[Theorems~1.3 and~1.4]{WY25}. Its dependence on $\xi$ will be larger than
that of our final coreset, but its row count is independent of the input row
count $N$.

We also use their fast bicriteria construction for the low-rank
split~\cite[Lemma~2.3]{WY25}.

\begin{lemma}[Woodruff--Yasuda bicriteria subspace]
\label{lem:WY-bicriteria}
Fix $1\le p<2$. Given $A\in\R^{N\times d}$, $1\le k<d$, and
$\delta\in(0,1/2)$, there is a randomized algorithm that returns a set
$\mathcal J$ of at most
\[
 C_pk\left(\log\frac{C_pN}{\delta}\right)^{C_p}
\]
rows whose span $F_0$ satisfies
\[
 \cost_A(F_0)\le\Gamma_0\OPT_A
\]
with probability at least $1-\delta$, where $\Gamma_0$ depends only on $p$.
The algorithm runs in
$\wtO_p(\nnz(A)+|\mathcal J|^\omega)$ time.
\end{lemma}

\begin{lemma}[Composition of strong row coresets]
\label{lem:coreset-composition}
Suppose $A_0=S_0A$ is a $(1\pm\xi_0)$ strong row coreset for $A$, and
$A_1=S_1A_0$ is a $(1\pm\xi_1)$ strong row coreset for $A_0$. Then, for every
$F\in\Fset_k$,
\[
 (1-\xi_0)(1-\xi_1)\cost_A(F)
 \le \cost_{A_1}(F)
 \le (1+\xi_0)(1+\xi_1)\cost_A(F).
\]
In particular, if $0<\eps<1/2$ and
$\xi_0=\xi_1=\eps/3$, then $A_1$ is a $(1\pm\eps)$ strong row coreset for
$A$. Moreover, $S_1S_0$ is again a row-sampling matrix for the original rows
of $A$.
\end{lemma}

\begin{proof}
Apply the two coreset inequalities successively. For the stated choice of
parameters,
$(1-\eps/3)^2\ge1-\eps$ and
$(1+\eps/3)^2\le1+\eps$ when $\eps<1/2$.
The final assertion follows because a product of row-selection and rescaling
maps still selects and rescales rows of the original matrix.
\end{proof}

\subsection{Leverage scores and Lewis weights}

Throughout this subsection, zero rows are deleted before computing Lewis
weights and restored with weight zero afterward.

\begin{definition}[Leverage scores]
Let $M\in\R^{N\times r}$ have rows $m_i^\top$. The leverage score of row
$i$ is
\[
    \tau_i(M):=m_i^\top(M^\top M)^\dagger m_i.
\]
Thus $\sum_i\tau_i(M)=\rank(M)$. If $M$ has full column rank, the
pseudoinverse may be replaced by the inverse, and
$\tau_i(MR)=\tau_i(M)$ for every invertible change of coordinates
$R\in\R^{r\times r}$.
\end{definition}

\begin{definition}[$\ell_p$ Lewis weights]\label{def:lewis}
Let $M\in\R^{N\times r}$ have full column rank and rows $m_i^\top$. The
$\ell_p$ Lewis weights of $M$ are the positive numbers
$w_1,\ldots,w_N$ satisfying
\begin{equation}\label{eq:lewis-definition}
    w_i=\tau_i\!\left(W^{1/2-1/p}M\right),
    \qquad W=\diag(w_1,\ldots,w_N).
\end{equation}
Equivalently,
\begin{equation}\label{eq:lewis-fixed-point}
    m_i^\top\!\left(M^\top W^{1-2/p}M\right)^{-1}m_i
    =w_i^{2/p}.
\end{equation}
We call $\widetilde w_i$ an $\alpha$-approximation to the Lewis weights if
$\alpha^{-1}w_i\le\widetilde w_i\le\alpha w_i$ for every positive-weight
row.
\end{definition}

Thus Lewis weights are ordinary leverage scores after a self-consistent
row reweighting; when $p=2$, they reduce to the usual leverage scores of
$M$.
For $1\le p<4$, Cohen and Peng prove that these weights exist and are
unique~\cite[Sec.~3, in particular Cor.~3.4]{CP15}. Since they are leverage
scores of a reweighted matrix,
\begin{equation}\label{eq:lewis-total}
    \sum_iw_i=\rank(M)=r.
\end{equation}
They are invariant under invertible changes of column coordinates:
$M$ and $MR$ have the same Lewis weights whenever $R$ is invertible.
Accordingly, for vectors lying in an $r$-dimensional subspace, their Lewis
weights mean the weights of their coordinates in any basis of that
subspace; this is independent of the chosen basis.

\begin{theorem}[Cohen--Peng Lewis-weight sampling]\label{thm:cp-sampling}
Fix $1\le p<2$ and $0<\eps<1$. Let $M\in\R^{N\times r}$ have Lewis
weights $w_i$, and
let $q$ be a distribution satisfying
$q_i\gtrsim_p w_i/r$ whenever $w_i>0$. If $\Smat$ is formed from
$m=\wtO_p(r\eps^{-2})$ i.i.d. samples from $q$ as in
\cref{def:iid-sampling}, then, with constant probability,
\[
    \norm{\Smat Mx}_p^p=(1\pm\eps)\norm{Mx}_p^p
    \qquad\text{simultaneously for every }x\in\R^r.
\]
\end{theorem}

This is the standard Lewis-weight subspace-embedding theorem of Cohen and
Peng~\cite[Theorem~7.1 and Lemma~7.4]{CP15}; logarithmic factors depend on the
fixed $p$. We later prove a vector-valued refinement tailored to subspace
costs.

\begin{algorithm}[htbp]
\caption{Iterative Algorithm to Compute the $\ell_p$ Lewis Weights~\cite{CP15}}
\label{alg:cp-iteration}
\begin{algorithmic}[1]
\State Initialize $w=\mathbf 1\in\R^N$
\For{$t=1,2,\ldots,T$}
    \State Let $W=\diag(w_1,\ldots,w_N)$
    \State Let $\tau\in\R^N$ be a constant-factor approximation of the
      leverage scores of $W^{1/2-1/p}M$
    \State Set $w_i\gets(w_i^{2/p-1}\tau_i)^{p/2}$
\EndFor
\State \Return $w$
\end{algorithmic}
\end{algorithm}

\begin{lemma}[Cohen--Peng iteration~\cite{CP15}]\label{lem:cp-iteration}
Fix $1\le p<2$, and let $M\in\R^{N\times r}$ have full column rank. After
$T=O_p(\log\log(N+2))$ iterations of
\cref{alg:cp-iteration}, the returned vector $w$ is a constant-factor
approximation to the $\ell_p$ Lewis weights of $M$.
\end{lemma}

\begin{lemma}[Lewis normalization]
\label{lem:lewis-position}
Let $M\in\R^{N\times r}$ have full column rank, let $1\le p\le2$, and let
$m_i^\top$ and $w_i$ be its rows and Lewis weights.
There exist an invertible matrix $J\in\R^{r\times r}$ and unit vectors
$y_i\in\R^r$ such that
\begin{equation}\label{eq:lewis-factorization}
 m_i=w_i^{1/p}Jy_i,\qquad
 \sum_iw_iy_iy_i^\top=I_r,\qquad
 \sum_iw_i=r.
\end{equation}
\end{lemma}

\begin{proof}
By the convention at the beginning of this subsection, $M$ has no zero
rows and every $w_i$ is positive. Set
\[
    G=M^\top W^{1-2/p}M,
    \qquad J=G^{1/2},
    \qquad y_i=w_i^{-1/p}G^{-1/2}m_i.
\]
Since $M$ has full column rank and $W$ is positive diagonal, $G$ is positive
definite. Thus $G^{1/2}$ and $G^{-1/2}$ are well defined. By
\eqref{eq:lewis-fixed-point},
\[
    \norm{y_i}_2^2
    =w_i^{-2/p}m_i^\top G^{-1}m_i
    =1,
\]
and the definition of $y_i$ gives
$m_i=w_i^{1/p}Jy_i$. Moreover,
\[
    \sum_iw_iy_iy_i^\top
    =G^{-1/2}\left(
       \sum_iw_i^{1-2/p}m_im_i^\top
     \right)G^{-1/2}
    =G^{-1/2}GG^{-1/2}
    =I_r.
\]
Taking traces and using $\norm{y_i}_2=1$ gives
$\sum_iw_i=r$. This proves \eqref{eq:lewis-factorization}.
\end{proof}

The next vector-valued sensitivity estimate is the offline $p\le2$
specialization of the online bound in~\cite[Lemma~8.15]{WY23}. We include a
short proof because the offline Lewis normalization makes the argument
particularly direct.

\begin{lemma}[Lewis weights bound $(p,2)$-sensitivities]
\label{lem:lewis-p2-sensitivity}
Let $M\in\R^{N\times r}$ have full column rank, let $1\le p\le2$, and let
$m_i^\top$ and $w_i$ be its rows and Lewis weights. For every $s\ge1$ and
every matrix $T\in\R^{s\times r}$,
\begin{equation}\label{eq:lewis-sensitivity}
  \norm{T m_i}_2^p\le w_i\sum_j\norm{T m_j}_2^p.
\end{equation}
\end{lemma}

\begin{proof}
Use \cref{lem:lewis-position} to write
$m_j=w_j^{1/p}Jy_j$ with unit $y_j$ and
$\sum_jw_jy_jy_j^\top=I_r$, and set $S=TJ$. The claim is immediate if
$S=0$. Otherwise, since $p\le2$ and
$\norm{Sy_j}_2\le\norm{S}_{\mathrm{op}}$,
\[
 \sum_jw_j\norm{Sy_j}_2^p
 \ge \norm{S}_{\mathrm{op}}^{p-2}
      \sum_jw_j\norm{Sy_j}_2^2
 =\norm{S}_{\mathrm{op}}^{p-2}\norm{S}_F^2
 \ge\norm{S}_{\mathrm{op}}^p.
\]
The conclusion follows from
$\norm{Tm_i}_2^p=w_i\norm{Sy_i}_2^p$ and
$\sum_j\norm{Tm_j}_2^p=\sum_jw_j\norm{Sy_j}_2^p$.
\end{proof}

In particular, the usual scalar $\ell_p$ sensitivity of row $i$ satisfies
\begin{equation}\label{eq:scalar-lewis-sensitivity}
    \sigma_i^{(p)}(M)
    :=\sup_{x:Mx\ne0}
      \frac{|m_i^\top x|^p}{\norm{Mx}_p^p}
    \le w_i.
\end{equation}

\subsection{Euclidean sketches}

\begin{definition}[Oblivious subspace embedding]
A distribution over matrices $\Pi\in\R^{s\times d}$ is an
$(r,\zeta,\delta)$ oblivious subspace embedding (OSE) if, for every fixed
subspace $L\subseteq\R^d$ with $\dim L\le r$,
\[
 \Prob_\Pi\!\left[
   (1-\zeta)\norm{x}_2^2\le\norm{\Pi x}_2^2
   \le(1+\zeta)\norm{x}_2^2
   \ \text{for every }x\in L
 \right]\ge1-\delta.
\]
The subspace must be fixed independently of the draw of $\Pi$.
\end{definition}

We use the following standard Euclidean sketching facts. For any fixed
vectors $v_1,\ldots,v_N$ and $0<\zeta,\delta<1$, a
Johnson--Lindenstrauss map with
$O(\zeta^{-2}\log(N/\delta))$ rows simultaneously preserves all
$\norm{v_i}_2^2$ within a factor $1\pm\zeta$ with probability at least
$1-\delta$~\cite{JL84}. For a fixed $r$-dimensional subspace, sparse OSEs
with constant distortion have $\wtO(r)$ rows and polylogarithmic column
sparsity, including at inverse-polynomial failure probability~\cite{NN13}.
Consequently, multiplying a matrix by the sparse OSE takes input-sparsity
time up to logarithmic factors.

\subsection{Uniform convergence tools}

For a fixed query, a row sample gives an unbiased estimate of its cost.
A scalar concentration inequality controls its error. A strong coreset,
however, requires one sample to be accurate for every query, so we must
control the largest sampling error over a function class. We first record
the fixed-query estimate and then collect the uniform-convergence tools used
for this purpose. All function classes in our applications are countable,
finite, or replaced by a countable dense subclass before these tools are
applied. Thus the suprema below are measurable in every use.
Symmetrization and chaining also underlie Lewis-weight and sensitivity-sampling
analyses for $\ell_p$ objectives~\cite{CP15,MMWY22,MO24}. In the broader
coreset literature, pseudodimension-based uniform convergence is a standard
ingredient of general sensitivity-sampling frameworks~\cite{FL11,BFLSZ21}.

\paragraph{Concentration for one query.}
For a fixed positive-cost query, normalize each sampled contribution by the
true cost. A sensitivity bound then gives a random summand $X$ with
$0\le X\le K$ and $\E X=1$. Necessarily $K\ge1$, and
$\operatorname{Var}(X)\le\E X^2\le K$ and $|X-1|\le K$.
Bernstein's inequality~\cite[Corollary~2.11]{BLM13} therefore gives, for
independent copies $X_1,\ldots,X_m$ and $0<\eta\le1$,
\[
 \Prob\!\left[
   \left|\frac1m\sum_{s=1}^mX_s-1\right|>\eta
 \right]
 \le2\exp\!\left(-\frac{3m\eta^2}{8K}\right).
\]
A zero-cost query has zero contribution from every row and is preserved
exactly.

\paragraph{Empirical averages and symmetrization.}
Let $I_1,\ldots,I_m$ be independent samples from a distribution $q$ on the
row indices. For a function $f$ on the rows, write
\[
 P_mf=\frac1m\sum_{s=1}^m f(I_s),
 \qquad
 Pf=\E_{I\sim q}[f(I)].
\]
Thus $P_mf$ is the sampled average and $Pf$ is its expectation. For a
function class $\mathcal G$, its absolute Rademacher complexity is
\[
 \Rad_m^{\rm abs}(\mathcal G)
 =\E_{I,\eps}\sup_{g\in\mathcal G}
 \left|\frac1m\sum_{s=1}^m\eps_sg(I_s)\right|,
\]
where the $\eps_s$ are independent uniform signs. If the sample
$I_1,\ldots,I_m$ is fixed, then
$\Rad_m^{\rm abs}(\mathcal G\mid I_1,\ldots,I_m)$ denotes the same
quantity with expectation only over the signs. The symmetrization lemma
\cite[Lemma~2.3.1]{VW96} bounds the expected uniform sampling error by
this signed process:
\begin{equation}\label{eq:standard-symmetrization}
    \E\sup_{g\in\mathcal G}|(P_m-P)g|
    \le2\Rad_m^{\rm abs}(\mathcal G).
\end{equation}

\paragraph{Covering numbers and Dudley's bound.}
For a probability distribution $\nu$ on the row indices and real-valued
functions $f,g$ on these indices, write
\[
 \norm{f-g}_{L_2(\nu)}
 =\left(\sum_i\nu_i|f(i)-g(i)|^2\right)^{1/2}.
\]
Let $N(\delta,\mathcal G,L_2(\nu))$ be the minimum cardinality of a set
$\mathcal N\subseteq\mathcal G$ such that every $g\in\mathcal G$ satisfies
$\norm{g-\widetilde g}_{L_2(\nu)}\le\delta$ for some
$\widetilde g\in\mathcal N$. A $\delta$-packing is a subset of $\mathcal G$
whose distinct elements have $L_2(\nu)$-distance greater than $\delta$.
Every maximal $\delta$-packing is also a $\delta$-cover. For the empirical
distribution $P_m$, this distance becomes
\[
 \norm{f-g}_{L_2(P_m)}
 =\left(\frac1m\sum_{s=1}^m
       |f(I_s)-g(I_s)|^2\right)^{1/2}.
\]
Covering numbers at all scales can be combined by Dudley's chaining
argument. Applying the usual entropy integral
\cite[Corollary~13.2]{BLM13} to an $\alpha$-net and bounding the discarded
increments by Cauchy--Schwarz gives the following truncated form; see also
\cite[Lemma~A.5]{BFT17} for a direct chaining proof of the normalized
empirical Rademacher version. Set
\[
 \mathcal G^\pm=\mathcal G\cup(-\mathcal G)\cup\{0\},
 \qquad
 R_m=\sup_{g\in\mathcal G}\norm{g}_{L_2(P_m)}.
\]
If $R_m=0$, the conditional Rademacher complexity is zero. If $R_m>0$, then
for every $0<\alpha\le R_m$,
\begin{equation}\label{eq:standard-dudley}
 \Rad_m^{\rm abs}(\mathcal G\mid I_1,\ldots,I_m)
 \le
 4\alpha+\frac{C}{\sqrt m}
 \int_\alpha^{R_m}
 \sqrt{\log N(u,\mathcal G^\pm,L_2(P_m))}\,du.
\end{equation}
Here symmetrizing the class converts the absolute supremum into an ordinary
one, and adjoining zero makes $R_m$ an upper bound on the relevant covering
scales. The term $4\alpha$ is the error from stopping the chaining at scale
$\alpha$.

For real-valued classes, pseudodimension is the VC dimension of the class of
subgraphs~\cite{Hau92}. Equivalently, $\Pdim(\mathcal G)$ is the largest $v$
for which there are points
$x_1,\ldots,x_v$ and thresholds $t_1,\ldots,t_v$ such that, for every
$S\subseteq[v]$, some $g\in\mathcal G$ satisfies $g(x_j)>t_j$ for $j\in S$
and $g(x_j)\le t_j$ for $j\notin S$. The VC-subgraph covering theorem
\cite[Theorem~2.6.7]{VW96} implies that if $\mathcal G$ takes values in
$[0,1]$ and has pseudodimension $v\ge1$, then, for every probability
distribution $\nu$,
\begin{equation}\label{eq:pdim-entropy}
    \log N(\delta,\mathcal G,L_2(\nu))
    \le Cv\log(C/\delta),
    \qquad 0<\delta<1.
\end{equation}

\paragraph{Lipschitz composition.}
Our loss functions are often Lipschitz functions of several simpler scalar
coordinates. Fix the sample, and let $\{v_f\}$ be a class of functions from
the row indices to $\R^a$, with coordinate functions
$v_f=(v_{f,1},\ldots,v_{f,a})$. Suppose that
$\phi:\R^a\to\R$ satisfies $\phi(0)=0$ and
$|\phi(u)-\phi(v)|\le L\norm{u-v}_2$. Adjoin the zero map to the class if
necessary, and let $(\eps_{s,\ell})$ be independent uniform signs.
Maurer's vector-contraction inequality~\cite[Corollary~4]{Mau16}, followed
by subadditivity over the coordinates, gives
\begin{equation}\label{eq:vector-contraction}
\begin{split}
 &\E_\eps\sup_f
 \left|\frac1m\sum_{s=1}^m
       \eps_s\phi(v_f(I_s))\right|\\
 &\quad\le
 2\sqrt2L\,\E_{\eps_{s,\ell}}\sup_f
 \frac1m\sum_{s=1}^m\sum_{\ell=1}^a
 \eps_{s,\ell}v_{f,\ell}(I_s)\\
 &\quad\le
 2\sqrt2L\sum_{\ell=1}^a
 \E_\eps\sup_f
 \left|\frac1m\sum_{s=1}^m
       \eps_s v_{f,\ell}(I_s)\right|.
\end{split}
\end{equation}
The factor $2$ in the first inequality accounts for the absolute value;
Maurer's theorem gives the factor $\sqrt2$ and the doubly indexed signs.

\subsection{The Euclidean duality map}

For $x$ in a real Hilbert space, define
\[
  J_p(x)=
  \begin{cases}
    \norm{x}_2^{p-2}x,&x\ne0,\\
    0,&x=0.
  \end{cases}
\]
Thus
\[
    \ip{J_p(x)}{x}=\norm{x}_2^p,
    \qquad
    \norm{J_p(x)}_2=\norm{x}_2^{p-1}\quad(x\ne0).
\]
At $p=1$, this means $J_1(x)=x/\norm{x}_2$ off zero and $J_1(0)=0$.
If $g$ is a standard Gaussian in a finite-dimensional real Hilbert space,
then, with $\gamma_p=\E|N(0,1)|^p$,
\begin{equation}\label{eq:gaussian-projection-moment}
    \E_g|\ip{v}{g}|^p=\gamma_p\norm{v}_2^p.
\end{equation}

\section{Coresets from a Low-Rank Split for
\texorpdfstring{$1\le p<2$}{1 <= p < 2}}
\label{sec:small-p-split}

For any linear subspace $F\subseteq\R^d$, let $P_F$ be the orthogonal
projector onto $F$ and let $Q_F=I-P_F$. The query family $\Fset_k$ consists
of all such subspaces of dimension at most $k$, and the cost of
$F\in\Fset_k$ is
\[
    \cost_A(F)=\norm{AQ_F}_{p,2}^p
    =\sum_i\norm{Q_Fa_i}_2^p.
\]
Our goal is to construct a single weighted sample of the rows of $A$ that
preserves this quantity simultaneously for every $F\in\Fset_k$.

Let $F_0\subseteq\R^d$ be a bicriteria subspace satisfying
\[
    r:=\dim(F_0)=\widetilde O(k),
    \qquad
    \cost_A(F_0)\le\Gamma_0\OPT_A
\]
for a constant $\Gamma_0\ge1$. For each row $a_i$, we use a vector
$b_i\in F_0$ as a constant-factor approximation to its Euclidean projection
onto $F_0$. The sketched regression procedure in \cref{sec:algorithm}
computes these vectors without forming an ambient-dimensional projection and
guarantees
\[
    \norm{a_i-b_i}_2
    \le C\dist(a_i,F_0)
    =C\norm{Q_{F_0}a_i}_2
\]
for an absolute constant $C$. Let $e_i=a_i-b_i$, and let $B$ and $E$ have
rows $b_i$ and $e_i$, respectively. The construction also ensures that the
rows of $B$ span $F_0$, and hence
\[
    A=B+E,\qquad
    \rank(B)=r,\qquad
    R:=\norm{E}_{p,2}^p
    \le C^p\cost_A(F_0)
    \le\Gamma\OPT_A
\]
for a constant $\Gamma\ge1$. These are the only properties of the split used
in this section.

The sampling distribution combines one score for each component of
$A=B+E$. Let $w_i$ be the $\ell_p$ Lewis weight of the low-rank row $b_i$.
When $R>0$, let
$\rho_i=\norm{e_i}_2^p/R$ be the fraction of the residual mass contributed
by row $i$, and choose a distribution $q$ satisfying
\[
    q_i\gtrsim_{p,\Gamma}\frac{w_i}{k+r},
    \qquad
    q_i\gtrsim_{p,\Gamma}\rho_i.
\]
Draw $m=\widetilde O_{p,\Gamma}((k+r)\eps^{-2})$ indices
$I_1,\ldots,I_m$ independently from $q$ and retain row
$a_{I_s}$ with scaling $(mq_{I_s})^{-1/p}$. Intuitively, the Lewis-weight
score controls the contribution of the
low-rank component $Q_Fb_i$, while the residual score controls that of
$Q_Fe_i$. In \cref{sec:small-p-split}, we prove that this sampling rule
preserves all queries simultaneously. In \cref{sec:algorithm}, we construct
$F_0$, compute the regression fits $b_i$, and estimate both scores
efficiently.

\subsection{The low-rank-split theorem}
\label{sec:split-theorem}

Suppose that the decomposition $A=B+E$ satisfies
\begin{equation}\label{eq:split-assumption}
 r=\rank(B),\qquad
 R=\norm{E}_{p,2}^p\le\Gamma\,\OPT_A
\end{equation}
for some constant $\Gamma\ge1$.  Put
\[
 D=k+r.
\]

Let $w_i$ denote the Lewis weight of the nonzero row $b_i$, with $w_i=0$
when $b_i=0$. When $R>0$, put $\rho_i=\norm{e_i}_2^p/R$; when $R=0$,
put $\rho_i=0$.

\begin{theorem}[Low-rank-split coreset]\label{thm:split}
Assume $k,r\ge1$, and let $q$ be any probability distribution satisfying
\begin{equation}\label{eq:robust-dominance}
 q_i\gtrsim_{p,\Gamma}\frac{w_i}{D}+\rho_i
 \qquad (i\in[n]).
\end{equation}
Fix $0<\eps,\eta<1/2$, and form $S$ by drawing
$m=\widetilde O_{p,\Gamma}(D\eps^{-2}\eta^{-2})$ rows independently from $q$
and reweighting them as in \cref{def:iid-sampling}. Then, with probability at
least $1-\eta$,
\[
 \cost_{SA}(F)=(1\pm\eps)\cost_A(F)
 \qquad\text{simultaneously for every }F\in\Fset_k.
\]
The suppressed factors are polylogarithmic in
$n+D+\eps^{-1}+\eta^{-1}$.
\end{theorem}

The balanced sensitivity lemma below supplies a distribution satisfying
\eqref{eq:robust-dominance} and bounds each row's contribution to every
query. We then normalize the query loss and decompose it into four scalar
coordinate classes. Appendix~\ref{sec:auxiliary-process} proves empirical
process bounds for these classes with no dependence on $d$. In
Section~\ref{sec:split-proof}, we combine these estimates through
symmetrization and vector contraction to prove \cref{thm:split}.

If $R=0$, then $A=B$, and the conclusion follows directly from
\cref{thm:low-rank-lewis}. Hence assume $R>0$ throughout the remainder of the
proof.

\subsubsection{Balanced sensitivity}

\begin{lemma}[Balanced sensitivity for a low-rank split]
\label{lem:sensitivity}
Let $A=B+E$ satisfy \eqref{eq:split-assumption} with $R>0$, let
$\rho_i=\norm{e_i}_2^p/R$, and let $w_i$ be the Lewis weight of $b_i$, with
$w_i=0$ when $b_i=0$. With $D=k+r$, define
\begin{equation}\label{eq:balanced-score}
 s_i=w_i+D\rho_i.
\end{equation}
Then
\[
 \norm{Q_Fa_i}_2^p
 \le C_{p,\Gamma}s_i\sum_j\norm{Q_Fa_j}_2^p
 \qquad(F\in\Fset_k).
\]
Moreover, if $S=\sum_i s_i$ and $q_i=s_i/S$, then
\[
 S\le2D,\qquad
 q_i\ge\frac{w_i}{2D},\qquad
 q_i\ge\frac{\rho_i}{2}.
\]
\end{lemma}

\begin{proof}
For a query $F$, write
$C_F=\sum_i\norm{Q_Fa_i}_2^p$. Since
$\OPT_A\le C_F$, $R\le\Gamma\OPT_A$, and $Q_F$ is a
contraction,
\[
 \norm{BQ_F}_{p,2}
 \le\norm{AQ_F}_{p,2}+\norm{EQ_F}_{p,2}
 \le C_F^{1/p}+R^{1/p}
 \le (1+\Gamma^{1/p})C_F^{1/p}.
\]
\Cref{lem:lewis-p2-sensitivity} and the row-wise $p$-triangle inequality
give
\[
\norm{Q_Fa_i}_2^p
\le2^{p-1}\bigl(\norm{Q_Fb_i}_2^p+\norm{Q_Fe_i}_2^p\bigr)
\le2^{p-1}\bigl((1+\Gamma^{1/p})^pw_i+\Gamma\rho_i\bigr)C_F
\le C_{p,\Gamma}s_iC_F.
\]
Finally,
\[
 S=\sum_iw_i+D\sum_i\rho_i=r+D\le2D.
\]
The two lower bounds for $q_i$ follow directly from
\eqref{eq:balanced-score}.
\end{proof}

\subsubsection{Balanced normalization}

Fix a distribution $q$ satisfying \eqref{eq:robust-dominance}. There is a
constant $c=c_{p,\Gamma}>0$ such that
\begin{equation}\label{eq:q-bounds}
 q_i\ge c\frac{w_i}{D},\qquad
 q_i\ge c\rho_i.
\end{equation}
For a query $F$ with positive cost, abbreviate $Q=Q_F$ and write
\[
 C_F:=\cost_A(F)=\sum_i\norm{Qa_i}_2^p.
\]
We normalize the low-rank and residual contributions of each row by
$q_iC_F$.

The i.i.d. rescaling by $q_i^{-1/p}$ turns the Lewis scale $w_i^{1/p}$ and
the residual scale $\rho_i^{1/p}$ into the row-wise factors
\[
 \alpha_i=(w_i/q_i)^{1/p},\qquad
 \beta_i=(\rho_i/q_i)^{1/p},
\]
with zero-over-zero values set to zero.
The bounds in \eqref{eq:q-bounds} imply
\begin{equation}\label{eq:alpha-beta-bounds}
 \alpha_i^p\le c^{-1}D,\qquad
 \beta_i^p\le c^{-1}.
\end{equation}

Define the normalized residual scale
\[
 t_F:=\left(\frac{R}{C_F}\right)^{1/p}.
\]
Since $C_F\ge\OPT_A$ and $R\le\Gamma\OPT_A$, we have
$0<t_F\le\Gamma^{1/p}$. Let $L_\Gamma=1+\Gamma^{1/p}$. By the definition
of $C_F$ and the fact that orthogonal projection is nonexpansive,
\[
 \norm{AQ}_{p,2}=C_F^{1/p},
 \qquad
 \norm{EQ}_{p,2}\le\norm{E}_{p,2}
 =R^{1/p}=t_FC_F^{1/p}.
\]
Since $B=A-E$, it follows that
\begin{equation}\label{eq:low-rank-total}
 \norm{BQ}_{p,2}
 \le\norm{AQ}_{p,2}+\norm{EQ}_{p,2}
 \le (1+t_F)C_F^{1/p}
 \le L_\Gamma C_F^{1/p}.
\end{equation}

By coordinate invariance of Lewis weights and
\cref{lem:lewis-position}, there exist a matrix $J\in\R^{d\times r}$ whose
columns span $F_0$ and vectors $y_i\in\R^r$ satisfying
\begin{equation}\label{eq:canonical-lewis}
  b_i=w_i^{1/p}Jy_i,\qquad \norm{y_i}_2\le1,\qquad
  \sum_iw_iy_iy_i^\top=I_r,\qquad \sum_iw_i=r.
\end{equation}
Here $\norm{y_i}_2=1$ when $b_i\ne0$, while $w_i=0$ and $y_i=0$ when
$b_i=0$. Likewise, write
\[
 e_i=R^{1/p}\rho_i^{1/p}u_i,
\]
where $u_i$ is a unit vector whenever $\rho_i>0$ and may be chosen to be any
unit vector otherwise. With these coordinates, define
\begin{equation}\label{eq:xy}
 x_F(i)=\alpha_i C_F^{-1/p}QJy_i
 =q_i^{-1/p}C_F^{-1/p}Qb_i,
 \qquad
 y_F(i)=t_F\beta_iQu_i
 =q_i^{-1/p}C_F^{-1/p}Qe_i.
\end{equation}
Consequently,
\begin{equation}\label{eq:Z}
 Z_F(i)=\norm{x_F(i)+y_F(i)}_2^p
 =\frac{\norm{Qa_i}_2^p}{q_iC_F},
 \qquad \E_qZ_F=1,
 \qquad 0\le Z_F(i)\le C_{p,\Gamma,c}D.
\end{equation}
The last inequality follows by combining \cref{lem:sensitivity} with
\eqref{eq:q-bounds}. Thus the remainder of the proof uses the sampling
distribution only through \eqref{eq:q-bounds}.

Finally, Equations~\eqref{eq:canonical-lewis}
and~\eqref{eq:low-rank-total} give
\[
 \sum_iw_i\norm{C_F^{-1/p}QJy_i}_2^p
 =\frac{\norm{BQ}_{p,2}^p}{C_F}\le L_\Gamma^p
\]
and applying \cref{lem:lewis-p2-sensitivity} to the rows of $B$ gives, for
every $i$ with $w_i>0$,
\begin{equation}\label{eq:normalized-pointwise}
 \norm{C_F^{-1/p}QJy_i}_2^p
 =\frac{\norm{Qb_i}_2^p}{w_iC_F}
 \le\frac{\norm{BQ}_{p,2}^p}{C_F}
 \le L_\Gamma^p.
\end{equation}

\subsubsection{An exact four-coordinate decomposition}
\label{sec:four-coordinate-decomposition}

By \eqref{eq:Z}, the normalized loss is the $p$th power of
$\norm{x_F(i)+y_F(i)}_2$. Fix $F$ and $i$, and abbreviate
$x=x_F(i)$ and $y=y_F(i)$. We first split $y$ into its component parallel
to $x$ and its component orthogonal to $x$. If $x\ne0$, set
\[
 \widehat x=\frac{x}{\norm{x}_2},
 \qquad
 c_F(i)=\ip{\widehat x}{y},
 \qquad
 y^\perp=y-c_F(i)\widehat x.
\]
If $x=0$, set $\widehat x=0$, $c_F(i)=0$, and $y^\perp=y$. In either case,
$y^\perp$ is orthogonal to $x$ and
\[
 x+y=(\norm{x}_2+c_F(i))\widehat x+y^\perp.
\]
Consequently,
\begin{equation}\label{eq:parallel-orthogonal}
 \norm{x_F(i)+y_F(i)}_2^2
 =\left|\norm{x_F(i)}_2+c_F(i)\right|^2+\norm{y^\perp}_2^2.
\end{equation}
Taking the $p/2$ power in \eqref{eq:parallel-orthogonal} and using
\eqref{eq:Z} gives
\begin{equation}\label{eq:four-coordinate-decoder}
 Z_F(i)=
 \left(
  \left|\norm{x_F(i)}_2+c_F(i)\right|^2
  +\norm{y^\perp}_2^2
 \right)^{p/2}.
\end{equation}

The parallel contribution depends on the sign of $c_F(i)$. We therefore
record its two magnitudes, their signed interaction, and the orthogonal
residual power. Recall that $J_p(z)=\norm{z}_2^{p-2}z$ for $z\ne0$ and
$J_p(0)=0$. Define
\begin{equation}\label{eq:four-coordinates}
\begin{aligned}
 T_F(i)&=\norm{x_F(i)}_2^p,
 &V_F(i)&=|c_F(i)|^p,\\
 H_F(i)&=p\ip{J_p(x_F(i))}{y_F(i)},
 &W_F^\perp(i)&=\dist(y_F(i),\operatorname{span}\{x_F(i)\})^p.
\end{aligned}
\end{equation}
When $x_F(i)\ne0$,
\[
 H_F(i)
 =p\,\sgn(c_F(i))\norm{x_F(i)}_2^{p-1}|c_F(i)|;
\]
when $x_F(i)=0$, our convention $J_p(0)=0$ gives $H_F(i)=0$.
Here $T_F$ is the low-rank power, $V_F$ is the parallel residual power,
$H_F$ records the signed interaction between the two, and $W_F^\perp$ is
the orthogonal residual power.

Apply \cref{lem:composition} with
$a=\norm{x_F(i)}_2$, $b=|c_F(i)|$, and
$\sigma=\sgn(c_F(i))$, taking $\sigma=1$ when $c_F(i)=0$.
It shows that the parallel power
$|\norm{x_F(i)}_2+c_F(i)|^p$ is a $C_p$-Lipschitz function of
$(T_F(i),V_F(i),H_F(i))$ on the set of feasible coordinate triples.
The outer $\ell_{2/p}$ norm is one-Lipschitz with respect to the
$\ell_1$ distance of the parallel and orthogonal powers. Consequently,
\begin{equation}\label{eq:coordinate-lipschitz}
 |Z_F(i)-Z_G(i)|
 \le C_p\bigl(
 |T_F(i)-T_G(i)|+|V_F(i)-V_G(i)|
 +|H_F(i)-H_G(i)|+|W_F^\perp(i)-W_G^\perp(i)|
 \bigr).
\end{equation}

We next bound the empirical complexity of each coordinate class separately.

\subsubsection{Control of the coordinate processes}

\paragraph{Low-rank power.}
For $\norm{BQ_F}_{p,2}>0$, set
\[
 h_F(i)=\frac{\norm{Q_Fb_i}_2^p}{q_i\norm{BQ_F}_{p,2}^p},
 \qquad
 \lambda_F=\frac{\norm{BQ_F}_{p,2}^p}{C_F}\le L_\Gamma^p.
\]
Then $T_F=\lambda_Fh_F$. By \eqref{eq:q-bounds},
\[
 K:=\max_{w_i>0}\frac{w_i}{q_i}\le c^{-1}D.
\]
Since $0\le\lambda_F\le L_\Gamma^p$,
\[
 \Rad_m^{\rm abs}(\{T_F\})
 \le L_\Gamma^p\Rad_m^{\rm abs}(\{h_F\}).
\]
Applying \cref{thm:low-rank-lewis} therefore gives
\begin{equation}\label{eq:T-rad}
 \Rad_m^{\rm abs}(\{T_F\})
 \le\widetilde O_{p,\Gamma,c}
 \left(\sqrt{D/m}+D/m\right).
\end{equation}
If $\norm{BQ_F}_{p,2}=0$, then $T_F$ is identically zero.

\paragraph{Parallel residual power.}
Equation~\eqref{eq:xy} gives
\[
 x_F(i)=q_i^{-1/p}C_F^{-1/p}Q_Fb_i,
 \qquad
 y_F(i)=t_F\beta_iQ_Fu_i.
\]
The scalar multiplying $Q_Fb_i$ in the first identity is positive. Hence,
when $Q_Fb_i\ne0$,
\[
 \frac{x_F(i)}{\norm{x_F(i)}_2}
 =\frac{Q_Fb_i}{\norm{Q_Fb_i}_2}\in F^\perp.
\]
Recalling the definition of $c_F(i)$ and using
$\ip{v}{Q_Fu_i}=\ip{v}{u_i}$ for $v\in F^\perp$, we obtain
\begin{align*}
 c_F(i)
 &=\ip{x_F(i)/\norm{x_F(i)}_2}{y_F(i)}\\
 &=t_F\beta_i
   \ip{Q_Fb_i/\norm{Q_Fb_i}_2}{Q_Fu_i}\\
 &=t_F\beta_i
   \ip{Q_Fb_i/\norm{Q_Fb_i}_2}{u_i}.
\end{align*}
This formula also holds when $Q_Fb_i=0$ if a normalized zero vector is
interpreted as zero. Since $b_i=w_i^{1/p}Jy_i$, define
\[
 g_F(i):=
 \ip{Q_FJy_i/\norm{Q_FJy_i}_2}{\beta_i u_i}
 =\ip{J_1(Q_FJy_i)}{\beta_i u_i}.
\]
Then $c_F=t_Fg_F$ and $V_F=t_F^p|g_F|^p$.
The vectors $\beta_i u_i$ have uniformly bounded norm by
\eqref{eq:alpha-beta-bounds}. Since $t_F^p\le\Gamma$,
\[
 \Rad_m^{\rm abs}(\{V_F\})
 \le\Gamma\Rad_m^{\rm abs}(\{|g_F|^p\}).
\]
Now \cref{lem:directional}, the truncated Dudley bound
\eqref{eq:standard-dudley}, and the Lipschitz map $g\mapsto|g|^p$ give
\begin{equation}\label{eq:V-rad}
 \Rad_m^{\rm abs}(\{V_F\})
 \le\widetilde O_{p,\Gamma,c}(\sqrt{r/m}).
\end{equation}

\paragraph{Signed interaction.}
Using \eqref{eq:xy} and the same projection identity as in the parallel
residual calculation gives
\begin{equation}
 H_F(i)
 =p\,t_F\ip{J_p(\alpha_iC_F^{-1/p}Q_FJy_i)}{\beta_i u_i}.
 \label{eq:H-phase}
\end{equation}

\Cref{lem:phase} controls classes of the form
\[
 i\longmapsto p\ip{J_p(\alpha_iAy_i)}{z_i},
 \qquad A:\R^r\to\cH.
\]
In \eqref{eq:H-phase}, the corresponding choices are
\[
 A=C_F^{-1/p}Q_FJ,
 \qquad z_i=\beta_i u_i.
\]
The bounds used in the preceding paragraph give
$\norm{z_i}_2\le c^{-1/p}$ and $t_F\le\Gamma^{1/p}$. Since $t_F$ is
independent of the row index, it contributes only a constant factor to the
absolute Rademacher supremum.

Let $\nu=m^{-1}\sum_{s=1}^m\delta_{I_s}$ be the empirical law of the
sampled indices. Recall that $P_mf=m^{-1}\sum_{s=1}^m f(I_s)$, and hence
$\sum_i\nu_i f(i)=P_mf$. Under the choices above, the empirical $p$-energy
in \cref{lem:phase} is
\begin{equation}\label{eq:empirical-energy}
 E_\nu=\sup_F\sum_i\nu_i
 \norm{\alpha_iC_F^{-1/p}Q_FJy_i}_2^p
 =\sup_F P_mT_F.
\end{equation}
Equations \eqref{eq:alpha-beta-bounds} and
\eqref{eq:normalized-pointwise} imply, pointwise in $i$ and $F$,
\[
 \norm{\alpha_iC_F^{-1/p}Q_FJy_i}_2^p
 \le C_{p,\Gamma,c}D.
\]
Thus $E_\nu\le C_{p,\Gamma,c}D$ deterministically. We first bound every
logarithmic factor involving $E_\nu$ by a polylogarithmic function of $D+2$
and absorb it into the $\widetilde O$ notation. Applying
\cref{cor:phase-rad} now gives, for $p>1$,
\begin{equation}\label{eq:H-conditional}
\Rad_m^{\rm abs}(\{H_F\}\mid I_1,\ldots,I_m)
 \le\widetilde O_{p,\Gamma,c}
 \left(\left(E_\nu^{\,1-1/p}
       +E_\nu^{\,(1-1/p)/2}\right)\sqrt{r/m}\right)
\end{equation}
At $p=1$, the endpoint statement of \cref{cor:phase-rad} directly gives
$\widetilde O_{1,\Gamma,c}(\sqrt{r/m})$.

For $p>1$, it remains to control the random empirical $p$-energy $E_\nu$ in
expectation.
For every query with
$\norm{BQ_F}_{p,2}>0$, the normalized low-rank contribution satisfies
$\E_qh_F=1$ and $T_F=\lambda_Fh_F$, where
$\lambda_F\le L_\Gamma^p$. If $\norm{BQ_F}_{p,2}=0$, then $T_F=0$.
Consequently,
\[
 E_\nu
 \le L_\Gamma^p\left(
  1+\sup_{\norm{BQ_F}_{p,2}>0}|P_mh_F-1|
 \right).
\]
Symmetrization and \cref{thm:low-rank-lewis} imply that, for
$m\ge\widetilde\Omega_{p,\Gamma,c}(D)$,
\[
 \E\sup_{\norm{BQ_F}_{p,2}>0}|P_mh_F-1|
 \le 2\Rad_m^{\rm abs}(\{h_F\})
 \le\widetilde O_{p,c}
 \left(\sqrt{D/m}+D/m\right)
 \le C_{p,\Gamma,c}.
\]
It follows that
\begin{equation}\label{eq:energy-expectation}
 \E E_\nu
 \le C_{p,\Gamma,c}.
\end{equation}
Both exponents of $E_\nu$ in \eqref{eq:H-conditional} lie in $(0,1)$.
After the deterministic treatment of the logarithmic factors above, Jensen's
inequality applies to these pure powers. Together with
\eqref{eq:energy-expectation}, it gives
\begin{equation}\label{eq:H-rad}
 \E\Rad_m^{\rm abs}(\{H_F\}\mid I_1,\ldots,I_m)
 \le\widetilde O_{p,\Gamma,c}(\sqrt{r/m}).
\end{equation}
Together with the endpoint estimate above, this proves \eqref{eq:H-rad} for
the full range $1\le p<2$.

\paragraph{Orthogonal residual power.}
By definition,
$W_F^\perp(i)=\dist(y_F(i),\operatorname{span}\{x_F(i)\})^p$.
Equation~\eqref{eq:xy} shows that
$\operatorname{span}\{x_F(i)\}=\operatorname{span}\{Q_Fb_i\}$ and
$y_F(i)=t_F\beta_iQ_Fu_i$. Therefore
\begin{equation}\label{eq:Wperp-distance}
 W_F^\perp(i)
 =t_F^p\beta_i^p
  \dist(Q_Fu_i,\operatorname{span}\{Q_Fb_i\})^p.
\end{equation}
The distance in \eqref{eq:Wperp-distance} can be written without projected
vectors:
\begin{align*}
 \dist(Q_Fu_i,\operatorname{span}\{Q_Fb_i\})
 &=\inf_{a\in\R}\norm{Q_F(u_i-a b_i)}_2\\
 &=\inf_{a\in\R,\,f\in F}\norm{u_i-a b_i-f}_2\\
 &=\dist(u_i,F+\operatorname{span}\{b_i\}).
\end{align*}
Consequently,
\[
 W_F^\perp(i)
 =t_F^p\beta_i^p
  \dist(u_i,F+\operatorname{span}\{b_i\})^p.
\]
After rescaling this class to have envelope at most one,
\cref{cor:angular} and the truncated Dudley bound
\eqref{eq:standard-dudley} give
\begin{equation}\label{eq:Wperp-rad}
 \Rad_m^{\rm abs}(\{W_F^\perp\})
 \le\widetilde O_{p,\Gamma,c}(\sqrt{k/m}).
\end{equation}
The $O_{p,\Gamma,c}(m^{-1})$ truncation terms for $V_F$ and $W_F^\perp$
are absorbed into their square-root bounds. The $D/m$ term in
\eqref{eq:T-rad} remains because the low-rank class has envelope of order
$D$.

\subsubsection{Proof of \texorpdfstring{\Cref{thm:split}}{the low-rank-split theorem}}
\label{sec:split-proof}

For $m\ge\widetilde\Omega_{p,\Gamma,c}(D)$, the symmetrization bound
\eqref{eq:standard-symmetrization}, with the measurability convention from
\cref{sec:prelim}, gives
\[
 \E\sup_{\substack{F\in\Fset_k\\ C_F>0}}|P_mZ_F-1|
 \le2\E_{I,\eps}\sup_{\substack{F\in\Fset_k\\ C_F>0}}
 \left|\frac1m\sum_s\eps_sZ_F(I_s)\right|.
\]
For each $F$ and $i$, collect the four coordinates into
\[
 v_F(i)=\bigl(T_F(i),V_F(i),H_F(i),W_F^\perp(i)\bigr)\in\R^4.
\]
On the set of feasible coordinate vectors, let $\phi$ be the decoder
satisfying $\phi(v_F(i))=Z_F(i)$. It maps the origin to zero, and
\eqref{eq:coordinate-lipschitz} gives
\[
 |\phi(u)-\phi(v)|
 \le C_p\norm{u-v}_1
 \le 2C_p\norm{u-v}_2.
\]
Extend $\phi$ to all of $\R^4$ if necessary without increasing its
Lipschitz constant. Applying Maurer's vector-contraction inequality
\eqref{eq:vector-contraction} and then separating the four coordinates gives
\begin{align*}
 &\E_{I,\eps}\sup_{\substack{F\in\Fset_k\\ C_F>0}}
 \left|\frac1m\sum_s\eps_sZ_F(I_s)\right|\\
 &\quad\le C_p\Bigl(
 \Rad_m^{\rm abs}(\{T_F\})
 +\Rad_m^{\rm abs}(\{V_F\})
 +\E_I\Rad_m^{\rm abs}(\{H_F\}\mid I_1,\ldots,I_m)
 +\Rad_m^{\rm abs}(\{W_F^\perp\})
 \Bigr).
\end{align*}
Substituting \eqref{eq:T-rad}, \eqref{eq:V-rad}, \eqref{eq:H-rad}, and
\eqref{eq:Wperp-rad}, and recalling that $D=k+r$, yields
\begin{equation}\label{eq:uniform-error}
 \E\sup_{\substack{F\in\Fset_k\\ C_F>0}}|P_mZ_F-1|
 \le\widetilde O_{p,\Gamma,c}
 \left(\sqrt{\frac{k+r}{m}}+\frac{k+r}{m}\right).
\end{equation}
For every fixed sample, the normalized loss $F\mapsto P_mZ_F$ is continuous
on each positive-cost stratum, so the countable-dense-subclass convention
indeed gives the displayed supremum. Zero-cost queries satisfy
$AQ_F=SAQ_F=0$ and are handled separately.

Taking $m=\widetilde O_{p,\Gamma}(D\eps^{-2}\eta^{-2})$ makes the
right-hand side of \eqref{eq:uniform-error} at most $\eta\eps$ after
adjusting the hidden constant. Markov's inequality therefore gives, with
probability at least $1-\eta$,
\begin{equation}\label{eq:power-preservation}
 1-\eps\le P_mZ_F\le1+\eps
 \qquad(F\in\Fset_k).
\end{equation}
By \cref{def:iid-sampling}, the sampled matrix $SA$ has at most $m$
nonzero rows, and
\[
 P_mZ_F
 =\frac{\norm{SAQ_F}_{p,2}^p}{\norm{AQ_F}_{p,2}^p}
\]
for every positive-cost query, proving the required $1\pm\eps$ cost
guarantee.\hfill\qedsymbol

\section{Input-Sparsity-Time Construction for
\texorpdfstring{$1\le p<2$}{1 <= p < 2}}
\label{sec:algorithm}

In this section, we give an efficient implementation of the sampling procedure in
\cref{thm:split}.
In the proof of \cref{thm:main-small-p}, we apply it to the
Woodruff--Yasuda coreset $A_0$ from
\cref{thm:WY-presparsification}. Since $A_0$ has
$\operatorname{poly}_p(k/\eps)$ rows, the resulting logarithmic factors depend
only on $k$ and $\eps$, rather than on the original row count $n$.

\begin{proposition}[Implementing the low-rank-split sampling procedure]
\label{prop:small-p-implementation}
Fix $1\le p<2$. Given $A\in\R^{N\times d}$, $1\le k<d$, and
$0<\eps<1/2$, there is a randomized algorithm running in
\[
 \wtO_p\!\left(\nnz(A)+k^\omega+k\eps^{-2}\right)
\]
time that returns a $(1\pm\eps)$ strong row coreset with probability at least
$5/6$. Its support size is at most
\begin{equation}\label{eq:small-p-N-bound}
 C_pk\eps^{-2}
 \left(\log\frac{C_p(N+k)}{\eps}\right)^{C_p}.
\end{equation}
The suppressed factors in the running time are polylogarithmic in
$N+d+\eps^{-1}$.
\end{proposition}

\begin{proof}
We describe the construction and indicate where each logarithmic factor
enters. The claim is immediate when $A=0$, so assume otherwise. Apply
\cref{lem:WY-bicriteria} with failure probability $1/24$. It returns a set
$\mathcal J$ of
\[
 m_0\le C_pk\log^{C_p}(N+2)
\]
rows whose span $F_0$ has dimension $r_0\le m_0$ and satisfies
\begin{equation}\label{eq:implementation-bicriteria}
 \cost_A(F_0)\le\Gamma_0\OPT_A
\end{equation}
for a constant $\Gamma_0$ depending only on $p$. The construction takes
$\wtO_p(\nnz(A)+m_0^\omega)$ time.

We next construct the split used in \cref{thm:split}. Draw a sparse right OSE
$\Pi\in\R^{s\times d}$ with $s=\wtO(m_0)$ and inverse-polynomial failure
probability. Each of the $N$ subspaces
\[
 F_0+\operatorname{span}\{a_i\},
 \qquad i\in[N],
\]
has dimension at most $r_0+1$, where $a_i^\top$ denotes the $i$th row of
$A$. A union bound therefore shows that, with high probability, $\Pi$ has
constant distortion on all of them simultaneously.
Condition on this event, and let the columns of
$V\in\R^{s\times r_0}$ form an orthonormal basis for $\Pi F_0$. Define
$b_i$ by the sketched least-squares regression
\[
 b_i=\operatorname*{argmin}_{b\in F_0}\norm{\Pi(a_i-b)}_2,
 \qquad e_i=a_i-b_i.
\]
The OSE is injective on $F_0$, so the minimizer is unique. Since $VV^\top$
is the orthogonal projector onto $\Pi F_0$, the regression condition is
equivalently
\begin{equation}\label{eq:implementation-fit}
 \Pi b_i=VV^\top\Pi a_i.
\end{equation}
If $b_i^*=P_{F_0}a_i$, sketched optimality and the OSE inequalities give
\[
 \norm{e_i}_2
 \le C\norm{a_i-b_i^*}_2.
\]
Consequently, for the matrices $B,E$ with rows $b_i,e_i$,
\[
 A=B+E,\qquad \rank(B)=r_0,\qquad
 \norm{E}_{p,2}^p\le C_p\cost_A(F_0)\le\Gamma_1\OPT_A.
\]
The equality $\rank(B)=r_0$ follows because the selected rows spanning $F_0$
are fitted exactly.

It remains to estimate the two parts of the balanced score. Set
$H=A\Pi^\top$. By \eqref{eq:implementation-fit}, the $i$th row of
$H(I-VV^\top)$ is
\[
 a_i^\top\Pi^\top(I-VV^\top)
 =\bigl((I-VV^\top)\Pi a_i\bigr)^\top
 =(\Pi e_i)^\top.
\]
Since $e_i\in F_0+\operatorname{span}\{a_i\}$, the OSE event gives
$c\norm{e_i}_2\le\norm{\Pi e_i}_2\le C\norm{e_i}_2$ for every $i$, where
$c,C>0$ are absolute constants. Applying a
Johnson--Lindenstrauss map on the right preserves the norms of these $N$
rows simultaneously up to a constant factor. Raising the resulting row-norm
estimates to the $p$th power gives values $\widetilde r_i$ satisfying
\begin{equation}\label{eq:implementation-residual-estimates}
 c_p\norm{e_i}_2^p\le \widetilde r_i\le C_p\norm{e_i}_2^p
 \qquad\text{for every }i\in[N].
\end{equation}
We next compute the Lewis weights of $B$ in $r_0$ dimensions. Consider the
coordinate map
\[
 T:F_0\longrightarrow\R^{r_0},
 \qquad T(x)=V^\top\Pi x.
\]
This map is invertible because $\Pi$ is injective on $F_0$ and the columns of
$V$ form a basis for $\Pi F_0$. Moreover,
\[
 T(b_i)=V^\top\Pi b_i=V^\top\Pi a_i=(HV)_{i,*}^\top.
\]
Thus the matrix
\[
 C=HV\in\R^{N\times r_0}
\]
has $T(b_i)^\top$ as its $i$th row. Since the rows $b_i$ span $F_0$ and $T$
is invertible, $C$ has full column rank. By the coordinate invariance of
Lewis weights in \cref{def:lewis}, $C$ and $B$ have the same Lewis weights.

Run the Cohen--Peng iteration from \cref{lem:cp-iteration} on $C$. At
iteration $t$, let $W_t$ be the current diagonal row-weight matrix. Apply the
input sparsity time leverage score algorithm (e.g.,~\cite{CW13}) to
$W_t^{1/2-1/p}C$. Using the factorization
$C=HV=A\Pi^\top V$, the routine can be implemented without forming the dense
matrix $C$. Together with
\cref{lem:cp-iteration}, this yields simultaneous constant-factor estimates
$\widetilde w_i$ of the Lewis weights. Fresh sketches and a union bound cover the
$O_p(\log\log(N+2))$ iteration calls.

Using the residual estimates from
\eqref{eq:implementation-residual-estimates}, set
\[
 \widetilde\rho_i
 =\frac{\widetilde r_i}{\sum_j\widetilde r_j},
 \qquad
 \widetilde q_i
 =\frac{\widetilde w_i+(k+r_0)\widetilde\rho_i}
 {\sum_j\bigl(\widetilde w_j+(k+r_0)\widetilde\rho_j\bigr)}
\]
when the residual mass is positive. If the residual mass is zero, put
$\widetilde\rho_i=0$ and normalize only the $\widetilde w_i$ in the definition
of $\widetilde q_i$. Constant-factor accuracy and
$\sum_iw_i=r_0$ imply
\[
 \widetilde q_i\gtrsim_p\frac{w_i}{k+r_0},
 \qquad
 \widetilde q_i\gtrsim_p\rho_i.
\]
Thus \cref{thm:split} applies with $\eta=1/12$.

Finally, $r_0\le C_pk\log^{C_p}(N+2)$. Expanding the logarithmic factors in
\cref{thm:split} and increasing $C_p$ gives the support bound
\eqref{eq:small-p-N-bound}. Assigning total failure probability $1/24$ to the
sketching routines, the bicriteria, sketching, and sampling events fail with
total probability at most $1/24+1/24+1/12=1/6$.

We next analyze the running time of the remaining steps. Applying the sparse
right OSE to $A$ takes $\wtO(\nnz(A))$ time, and computing the basis $V$ takes
$\wtO(m_0^\omega)$ time because $s=\wtO(m_0)$. Computing the residual
estimates in \eqref{eq:implementation-residual-estimates} with the
Johnson--Lindenstrauss sketch takes
$\wtO(\nnz(A)+m_0^\omega)$ time. Each Cohen--Peng iteration takes
$\wtO_p(\nnz(A)+m_0^\omega)$ time. The
$O_p(\log\log(N+2))$ iterations are absorbed by the logarithmic factors in
$\wtO_p$. Finally, forming the sampling distribution and drawing the
$\wtO_p(k\eps^{-2})$ weighted rows take
$\wtO_p(\nnz(A)+k\eps^{-2})$ time. Since
$m_0\le C_pk\log^{C_p}(N+2)$, the total running time is
\[
 \wtO_p\!\left(\nnz(A)+k^\omega+k\eps^{-2}\right),
\]
as claimed.
\end{proof}

\begin{proof}[Proof of \cref{thm:main-small-p}]
For the row bound independent of $n$, first apply
\cref{thm:WY-presparsification} to $A$ with accuracy $\eps/3$ and failure
probability $1/12$. Denote its output by $A_0=S_0A$ and its number of rows by
$N_0$. Equation~\eqref{eq:WY-presparsification-size} gives
\begin{equation}\label{eq:small-p-presparsified-size}
 N_0\le
 C_pk\eps^{-4/p}
 \left(\log\frac{C_pk}{\eps}\right)^{C_p}.
\end{equation}
Condition on this event and apply
\cref{prop:small-p-implementation} to $A_0$ with accuracy $\eps/3$. Since
\[
 \log\frac{C_p(N_0+k)}{\eps}
 \le C_p\log\frac{C_pk}{\eps},
\]
the resulting matrix $A_1=S_1A_0$ has at most
\[
 C_pk\eps^{-2}
 \left(\log\frac{C_pk}{\eps}\right)^{C_p}
\]
rows after adjusting $C_p$. By \cref{lem:coreset-composition},
$A_1=S_1S_0A$ is a $(1\pm\eps)$ strong row coreset for $A$. The two stages
fail with total probability at most $1/12+1/6=1/4$, which is below $1/3$.

Because $A_0$ is a weighted row subset of $A$,
$\nnz(A_0)\le\nnz(A)$, and $k^\omega\le d^\omega$. The total running time is
therefore
\[
 \wtO_p\!\left(\nnz(A)+d^\omega+k\eps^{-2}\right).\qedhere
\]
\end{proof}

\paragraph{Lower bound.}
For every fixed $1\le p<2$, the upper bound in
\cref{thm:main-small-p} is tight up to logarithmic factors whenever
$0<\eps<1/2$ and $k+1\ge C\log(1/\eps)$. Woodruff and Yasuda observed the
following reduction to sampling-based subspace embeddings
\cite[Sec.~1.3]{WY23}. Set $d_0=k+1$, and for a unit vector
$x\in\R^{d_0}$ let $F_x=x^\perp$. Then
\[
 \norm{B(I-P_{F_x})}_{p,2}^p=\norm{Bx}_p^p.
\]
Thus a strong row coreset for residual queries over all subspaces of rank at
most $k$ gives a for-all sampling-based $\ell_p$ embedding. The lower bound of
Li, Wang, and Woodruff~\cite[Theorem~6.1]{LWW21}, applied in dimension
$d_0=k+1$, therefore requires $\widetilde\Omega_p(k\eps^{-2})$ rows under the
stated conditions. Consequently, the worst-case strong row coreset size is
$\widetilde\Theta_p(k\eps^{-2})$ in this parameter range.

\section{One-Round Sampling for \texorpdfstring{$p>2$}{p > 2}}
\label{sec:large-p-one-round}

\subsection{Ridge leverage scores}

Let $\A\in\R^{N\times d}$, and let $\A_k$ be a best rank-$k$ approximation
to $\A$ in Frobenius norm. Set
\[
    \lambda_{\A}:=\frac{\norm{\A-\A_k}_F^2}{k}.
\]
The ridge leverage score of row $i$ is
\[
    \tau_i^{\lambda_{\A}}(\A)
    :=\arow_i^\top
      (\A^\top\A+\lambda_{\A}\I)^{\dagger}
      \arow_i,
\]
where $\dagger$ denotes the Moore--Penrose pseudoinverse.
When the matrix is clear, we abbreviate this score to $\tau_i$.
Following Woodruff and Yasuda~\cite{WY25}, we call sampling with weights
$\bigl(\tau_i^{\lambda_{\A}}(\A)\bigr)^{p/2}$ \emph{root ridge leverage
score sampling}.

\begin{lemma}[Ridge leverage score trace bound~\cite{CMM17}]
\label{lem:trace}
For $\lambda_{\A}=\norm{\A-\A_k}_F^2/k$,
\[
    \sum_{i=1}^N\tau_i^{\lambda_{\A}}(\A)\leq 2k.
\]
\end{lemma}

\begin{proof}
If $\lambda_{\A}=0$, then $\norm{\A-\A_k}_F=0$ and hence
$\operatorname{rank}(\A)\leq k$.
In this case,
\[
    \sum_{i=1}^N\tau_i^0(\A)
    =\operatorname{tr}\!\left(
      \A^\top\A(\A^\top\A)^\dagger
    \right)
    =\operatorname{rank}(\A)
    \leq k.
\]

Assume now that $\lambda_{\A}>0$.
Let $\sigma_1\geq\sigma_2\geq\cdots\geq 0$ be the singular values of $\A$.
Then
\[
\begin{aligned}
    \sum_{i=1}^N\tau_i^{\lambda_{\A}}(\A)
    &=\operatorname{tr}\!\left(
      \A^\top\A(\A^\top\A+\lambda_{\A}\I)^{\dagger}
      \right)\\
    &=\sum_j\frac{\sigma_j^2}{\sigma_j^2+\lambda_{\A}}.
\end{aligned}
\]
The first $k$ summands contribute at most $k$.
For the remaining summands,
\[
    \sum_{j>k}\frac{\sigma_j^2}{\sigma_j^2+\lambda_{\A}}
    \leq
    \frac{1}{\lambda_{\A}}\sum_{j>k}\sigma_j^2
    =k.
\]
\end{proof}

\subsection{The Woodruff--Yasuda one-round theorem}

We use the following normalization of the Woodruff--Yasuda one-round guarantee
for $p>2$~\cite[Theorem~3.3]{WY25}.

\begin{theorem}[Woodruff--Yasuda one-round theorem]
\label{thm:WY-one-round}
Fix $p>2$.
There are constants $c_p>0$ and $C_p\geq 1$ such that the following holds.
Let $\A\in\R^{N\times d}$ with $N\geq 3$, let
$\zeta,\eta\in(0,1/2)$, and suppose
\[
    N\geq \frac{C_p k^{p/2}}{\zeta}.
\]
Set
\[
    \alpha
    :=\frac{c_p\zeta^2}
    {(\log N)^3+\log(1/\eta)}.
\]
Let $\Smat$ be an $\ell_p$ sampling matrix whose probabilities satisfy
\[
    q_i\geq
    \min\left\{1,
    \frac{N^{p/2-1}
    \bigl(\tau_i^{\lambda_{\A}}(\A)\bigr)^{p/2}}
    {\alpha}
    \right\}.
\]
Then, with probability at least $1-\eta$, simultaneously for every
$F\in\Fset_k$,
\[
    \cost_{\Smat\A}(F)
    =(1\pm\zeta)\cost_{\A}(F)
    \pm
    C_p\zeta
    \left(\log\frac{C_pN}{\zeta\eta}\right)^{p/2}
    \OPT_A.
\]
\end{theorem}

\subsection{Uniform parameter tuning}

The next lemma tunes the internal parameter $\zeta$ so that both the
multiplicative and additive errors are at most a prescribed value $\theta$.
The upper bound $M\geq N$ allows the same tuning to be used throughout the
recursion.

\begin{lemma}[Uniform tuning]
\label{lem:tuning}
Fix $M\geq N\geq 3$ and $\theta,\eta\in(0,1/2)$.
Let
\[
    L:=\log\!\left(\frac{C_pM}{\theta\eta}\right),
    \qquad
    \zeta:=\frac{c_p\theta}{L^{p/2}},
\]
where $c_p>0$ is sufficiently small.
Suppose
\[
    N\geq C_p k^{p/2}\theta^{-1}L^{p/2}.
\]
Let $\Smat$ be obtained by applying \cref{thm:WY-one-round} with internal
accuracy $\zeta$ and failure probability $\eta$.
Then, with probability at least $1-\eta$, simultaneously for every
$F\in\Fset_k$,
\[
    \cost_{\Smat\A}(F)
    =(1\pm\theta)\cost_{\A}(F)
    \pm\theta\OPT_A.
\]
Moreover, the sampling threshold obeys
\[
    \alpha
    \geq
    \frac{c_p\theta^2}
    {L^p\bigl((\log M)^3+\log(1/\eta)\bigr)}.
\]
\end{lemma}

\begin{proof}
The assumed lower bound on $N$ implies the size condition in
\cref{thm:WY-one-round}, after changing $C_p$.
Also $\zeta\leq\theta$.
For the additive term, observe that
\[
\begin{aligned}
    \log\frac{C_pN}{\zeta\eta}
    &\leq
    \log\frac{C_pM}{\theta\eta}
    +\log\frac{\theta}{\zeta}\\
    &\leq
    L+\log\bigl(c_p^{-1}L^{p/2}\bigr)
    \leq C_pL.
\end{aligned}
\]
Thus, by taking $c_p$ sufficiently small,
\[
    C_p\zeta
    \left(\log\frac{C_pN}{\zeta\eta}\right)^{p/2}
    \leq \theta.
\]
Together with $\zeta\leq\theta$, this turns the one-round guarantee in
\cref{thm:WY-one-round} into
\[
    \cost_{\Smat\A}(F)
    =(1\pm\theta)\cost_{\A}(F)
    \pm\theta\OPT_A
\]
simultaneously for every $F\in\Fset_k$.
Finally,
\[
    \alpha
    =\frac{c_p\zeta^2}{(\log N)^3+\log(1/\eta)}
    \geq
    \frac{c_p\theta^2}
    {L^p\bigl((\log M)^3+\log(1/\eta)\bigr)},
\]
after another adjustment of $c_p$.
\end{proof}

\section{Threshold-aware row-count analysis}
\label{sec:threshold}

The retained row count is governed by a sum of truncated powers of ridge
leverage scores. The crucial point is that keeping the entire sampling scale
inside the truncation allows this sum to be bounded by the linear ridge
leverage score mass in \cref{lem:trace}.

\begin{lemma}[Threshold summation]
\label{lem:threshold}
Let $v>1$, $c>0$, and $\tau_1,\dots,\tau_N\geq 0$.
Then
\[
    \sum_{i=1}^N\min\{1,c\tau_i^v\}
    \leq c^{1/v}\sum_{i=1}^N\tau_i.
\]
In particular, if $p>2$, $v=p/2$, and $\sum_i\tau_i\leq 2k$,
\[
    \sum_{i=1}^N\min\{1,c\tau_i^{p/2}\}
    \leq 2k c^{2/p}.
\]
\end{lemma}

\begin{proof}
Set $y_i=c^{1/v}\tau_i$.
For every $y\geq 0$,
\[
    \min\{1,y^v\}\leq y.
\]
Indeed, if $0\leq y\leq 1$, then $y^v\leq y$, while if $y>1$, then
$\min\{1,y^v\}=1<y$.
Therefore
\[
    \min\{1,c\tau_i^v\}
    =\min\{1,y_i^v\}
    \leq y_i
    =c^{1/v}\tau_i.
\]
Summing over $i$ proves the claim.
Equality at $\tau_i=c^{-1/v}$ shows that the factor $c^{1/v}$ is best
possible.
\end{proof}

\begin{corollary}[Expected one-round row count]
\label{cor:expectation}
Fix $p>2$ and $1\le k<d$. Let $\A\in\R^{N\times d}$ and $\alpha>0$, and set
$\lambda_{\A}=\norm{\A-\A_k}_F^2/k$.
Let $\Smat$ be the $\ell_p$ sampling matrix associated with
\[
    q_i=
    \min\left\{1,
    \frac{N^{p/2-1}\tau_i^{p/2}}{\alpha}
    \right\},
\]
where $\tau_i=\tau_i^{\lambda_{\A}}(\A)$.
Then
\[
    \E[\nnz(\Smat)]
    \leq
    2k\alpha^{-2/p}N^{1-2/p}.
\]
\end{corollary}

\begin{proof}
Apply \cref{lem:threshold} with
$c=N^{p/2-1}/\alpha$ and use \cref{lem:trace}:
\[
\begin{aligned}
    \E[\nnz(\Smat)]
    &=\sum_{i=1}^Nq_i\\
    &\leq
    2k\left(\frac{N^{p/2-1}}{\alpha}\right)^{2/p}
    =2k\alpha^{-2/p}N^{1-2/p}.
\end{aligned}
\]
\end{proof}

To use the improved expectation bound in the recursive construction, we also
need a high-probability bound on the retained row count.

\begin{lemma}[Bernoulli row-count concentration]
\label{lem:chernoff}
Let $Z$ be a sum of independent Bernoulli random variables with mean $\mu$.
For every $\eta\in(0,1/2)$,
\[
    \Prob\!\left[Z>3\bigl(\mu+\log(1/\eta)\bigr)\right]
    \leq\eta.
\]
\end{lemma}

\begin{proof}
Set $L:=\log(1/\eta)$ and
\[
    t:=2\sqrt{\mu L}+2L.
\]
Bernstein's inequality gives
\[
    \Prob[Z-\mu\geq t]
    \leq
    \exp\!\left(-\frac{t^2}{2(\mu+t/3)}\right)
    \leq e^{-L}
    =\eta.
\]
Here the second inequality follows from
$t^2\geq 2L(\mu+t/3)$.
Moreover, $2\sqrt{\mu L}\leq\mu+L$, and hence
\[
    \mu+t
    \leq 2\mu+3L
    \leq 3(\mu+L).
\]
The claimed bound follows.
\end{proof}

Combining \cref{cor:expectation,lem:chernoff}, one sampling round on an
$N$-row matrix produces at most
\[
    O\!\left(
      k\alpha^{-2/p}N^{1-2/p}+\log(1/\eta)
    \right)
\]
rows with probability at least $1-\eta$.

\section{Recursive Construction for \texorpdfstring{$p>2$}{p > 2}}
\label{sec:construction}

We obtain the final coreset by applying the one-round sampler recursively to
the current weighted matrix. The recursion stops when the number of active
rows falls below a prescribed threshold or after a bounded number of rounds.
The next subsection analyzes the resulting row-count recurrence and shows
that the latter number of rounds suffices to reach its fixed-point bound. We
then apply the Woodruff--Yasuda pre-sparsification from
\cref{thm:WY-presparsification}. The reduced input size ensures that a single
recursive reduction achieves the desired row bound.

\subsection{Inner recursive sampling}
\label{sec:one-pass}

Fix an active matrix $\A^{(0)}\in\R^{N\times d}$ with $N\geq1$, a target
relative error $\rho\in(0,1/2)$, and a failure probability
$\eta\in(0,1/2)$.
Set
\[
    \gamma:=1-\frac{2}{p}\in(0,1)
\]
and
\[
    r:=\max\left\{1,
      \left\lceil
      \frac{\log\log(\max\{e^e,N\})}{-\log\gamma}
      \right\rceil
      \right\}.
\]
Choose
\[
    \theta:=\frac{c\rho}{r},
    \qquad
    \eta_0:=\frac{\eta}{4r},
\]
where $c>0$ is a sufficiently small absolute constant, and define
\[
    L:=\log\!\left(\frac{C_pN}{\theta\eta_0}\right),
    \qquad
    \zeta:=\frac{c_p\theta}{L^{p/2}}.
\]
Finally, set the stopping threshold
\[
    N_{\min}:=\frac{C_pk^{p/2}}{\zeta}
    =\frac{C_p}{c_p}k^{p/2}\theta^{-1}L^{p/2}.
\]

The parameter choices imply, after increasing $C_p$ if necessary,
\begin{equation}
\label{eq:pass-parameter-relations}
    r\leq C_pL,
    \qquad
    \log N\leq L,
    \qquad
    \log(1/\eta_0)\leq L,
    \qquad
    L\leq C_p\log\!\left(\frac{C_pN}{\rho\eta}\right).
\end{equation}
These relations follow from $\theta\eta_0=c\rho\eta/(4r^2)$ and
$r=O_p(\log\log(\max\{e^e,N\}))$.

In round $t$, let $\A^{(t-1)}$ be the current weighted active matrix and let
$N_{t-1}$ be its number of rows.
If $N_{t-1}\leq N_{\min}$, stop and output the current matrix.
Otherwise compute the ridge leverage scores of $\A^{(t-1)}$ at
\[
    \lambda_{t-1}
    :=\frac{\norm{\A^{(t-1)}-(\A^{(t-1)})_k}_F^2}{k},
\]
set
\[
    \alpha_t
    :=\frac{c_p\zeta^2}
    {(\log N_{t-1})^3+\log(1/\eta_0)},
\]
and sample independently with
\[
    q_i^{(t)}
    :=\min\left\{1,
    \frac{N_{t-1}^{p/2-1}
    \bigl(\tau_i^{\lambda_{t-1}}(\A^{(t-1)})\bigr)^{p/2}}
    {\alpha_t}
    \right\}.
\]
Let $\A^{(t)}=\Smat_t\A^{(t-1)}$ after deleting zero rows.
The sampler only deletes rows, so $N_t\le N_{t-1}\le N$ in every round.
The procedure executes at most $r$ rounds.

\subsection{Row-count contraction and its fixed point}

We first isolate the deterministic fixed-point calculation. We then combine
it with the high-probability row bound from \cref{sec:threshold}.

\begin{lemma}[Fixed-point recurrence; cf.~\cite{MMWY22}]
\label{lem:fixed-point}
Let $\gamma\in(0,1)$ and let $N_0\geq 3$.
Suppose
\[
    N_t\leq K N_{t-1}^{\gamma}
    \qquad (t=1,\dots,r),
\]
where $K\geq 1$ and
\[
    r\geq
    \frac{\log\log N_0}{-\log\gamma}.
\]
Then
\[
    N_r\leq eK^{1/(1-\gamma)}.
\]
In particular, when $\gamma=1-2/p$,
\[
    N_r\leq eK^{p/2}.
\]
\end{lemma}

\begin{proof}
If some $N_t=0$, the conclusion is immediate.
Otherwise set $a_t:=\log N_t$.
Then
\[
    a_t\leq \gamma a_{t-1}+\log K.
\]
Unrolling gives
\[
    a_r
    \leq
    \gamma^r\log N_0
    +(\log K)\sum_{j=0}^{r-1}\gamma^j
    \leq
    1+\frac{\log K}{1-\gamma}.
\]
Exponentiating proves the claim.
\end{proof}

\begin{proposition}[High-probability row bound]
\label{prop:one-pass-size}
Fix $p>2$, let $\A^{(0)}\in\R^{N\times d}$ with $N\ge1$ and
$1\le k<d$, and let $\rho,\eta\in(0,1/2)$. Let the recursive procedure and the
parameter $L$ be as defined in \cref{sec:one-pass}. With probability at least
$1-\eta/2$, the output has at most
\[
    C_p\frac{k^{p/2}}{\rho^2}L^{p+5}
\]
rows.
\end{proposition}

\begin{proof}
The definition of $L$ is the common tuning scale in \cref{lem:tuning} with
$M=N$, target error $\theta$, and failure probability $\eta_0$.
Since $N_{t-1}\leq N$, the thresholds satisfy
\[
    \alpha_t\geq\alpha_*:=
    \frac{c_p\theta^2}
    {L^p\bigl((\log N)^3+\log(1/\eta_0)\bigr)}.
\]
Conditioned on all previous rounds, \cref{cor:expectation} gives
\[
    \mu_t:=\E[N_t\mid \A^{(t-1)}]
    \leq
    2k\alpha_t^{-2/p}N_{t-1}^{1-2/p}.
\]
By \cref{lem:chernoff}, except with conditional probability at most $\eta_0$,
\[
\begin{aligned}
    N_t
    &\leq
    C_p k\alpha_t^{-2/p}N_{t-1}^{1-2/p}
    +C_p\log(1/\eta_0)\\
    &\leq
    K N_{t-1}^{1-2/p},
\end{aligned}
\]
where
\[
    K:=C_p\left(
      k\alpha_*^{-2/p}+\log(1/\eta_0)
      \right).
\]
After increasing $C_p$, we have $K\geq1$. The last inequality uses
$N_{t-1}^{1-2/p}\geq 1$.
A union bound shows that this recurrence holds in every executed round with
probability at least $1-r\eta_0\geq 1-\eta/4$.

If the algorithm stops early, the output has at most
\[
    N_{\min}
    \leq C_pk^{p/2}\theta^{-1}L^{p/2}
    \leq
    C_p\frac{k^{p/2}}{\rho^2}L^{p+5}
\]
rows, using \cref{eq:pass-parameter-relations} and $\rho<1$.
Otherwise, the algorithm executes all $r$ rounds. In this case
$N>N_{\min}\geq3$, and the definition of $r$ gives
\[
    r\geq\frac{\log\log N}{-\log(1-2/p)}.
\]
Thus all hypotheses of \cref{lem:fixed-point} hold, and it yields
\[
\begin{aligned}
    N_r
    &\leq eK^{p/2}\\
    &\leq
    C_p\left(
    k^{p/2}\alpha_*^{-1}
    +\bigl(\log(1/\eta_0)\bigr)^{p/2}
    \right).
\end{aligned}
\]
Now
\[
\begin{aligned}
    \alpha_*^{-1}
    &\leq
    C_p\theta^{-2}L^p
    \bigl((\log N)^3+\log(1/\eta_0)\bigr)\\
    &\leq
    C_p\rho^{-2}r^2L^{p+3}\\
    &\leq
    C_p\rho^{-2}L^{p+5}.
\end{aligned}
\]
The last two inequalities use $\theta=c\rho/r$ and
\cref{eq:pass-parameter-relations}.
Also $\bigl(\log(1/\eta_0)\bigr)^{p/2}\leq L^{p/2}$, which is absorbed by
the same bound because $k\geq1$ and $\rho<1$.
This proves the claim after adjusting constants.
\end{proof}

\subsection{Accumulation of the approximation error}

The row-count analysis does not by itself control the distortion accumulated
across the recursive rounds. We next show that the mixed
additive-multiplicative guarantees from the executed rounds compose to a
multiplicative guarantee.

For $t\geq0$, write
\[
    C_t(F):=\cost_{\A^{(t)}}(F),
    \qquad
    \OPT_t:=\min_{F\in\Fset_k}C_t(F).
\]

\begin{lemma}[Composition of additive-multiplicative rounds]
\label{lem:error-composition}
Let $\theta\in(0,1/2)$ and let $T\geq0$ be an integer.
Suppose that for every $t=1,\dots,T$ and every $F\in\Fset_k$,
\[
    C_t(F)=(1\pm\theta)C_{t-1}(F)\pm\theta\OPT_{t-1}.
\]
If $T\theta\leq c_0$ for a sufficiently small absolute constant $c_0$, then
\[
    C_T(F)=(1\pm C T\theta)C_0(F)
    \qquad\text{for every }F\in\Fset_k,
\]
where $C$ is an absolute constant.
\end{lemma}

\begin{proof}
Let $F^*$ minimize $C_0$.
Since $\OPT_{t-1}\leq C_{t-1}(F^*)$,
\[
    C_t(F^*)
    \leq
    (1+\theta)C_{t-1}(F^*)+\theta\OPT_{t-1}
    \leq
    (1+2\theta)C_{t-1}(F^*).
\]
Hence
\[
    \OPT_t\leq C_t(F^*)
    \leq(1+2\theta)^t\OPT_0
    \leq 2\OPT_0
\]
when $T\theta$ is sufficiently small.
For arbitrary $F$, unrolling the upper recurrence gives
\[
\begin{aligned}
    C_T(F)
    &\leq
    (1+\theta)^T C_0(F)
    +\theta\sum_{s=0}^{T-1}
      (1+\theta)^{T-1-s}\OPT_s\\
    &\leq
    e^{T\theta}C_0(F)
    +2T\theta e^{T\theta}\OPT_0\\
    &\leq
    (1+CT\theta)C_0(F),
\end{aligned}
\]
because $\OPT_0\leq C_0(F)$.
Similarly,
\[
\begin{aligned}
    C_T(F)
    &\geq
    (1-\theta)^T C_0(F)
    -\theta\sum_{s=0}^{T-1}
      (1-\theta)^{T-1-s}\OPT_s\\
    &\geq
    (1-T\theta)C_0(F)-2T\theta\OPT_0\\
    &\geq
    (1-3T\theta)C_0(F).
\end{aligned}
\]
This proves the claim.
\end{proof}

\subsection{The recursive reduction}

The row-count and approximation analyses above combine into the following
reduction. We will apply it once to the pre-sparsified matrix.

\begin{theorem}[Recursive reduction]
\label{thm:one-pass}
Fix $p>2$.
Given $\A\in\R^{N\times d}$ with $N\geq1$, $k\geq1$, and
$\rho,\eta\in(0,1/2)$, the recursive procedure above outputs a weighted row
subset $\A'$ such that, with probability at least $1-\eta$,
\[
    \cost_{\A'}(F)=(1\pm\rho)\cost_{\A}(F)
    \qquad\text{for every }F\in\Fset_k,
\]
and $\A'$ has at most
\[
    C_p\frac{k^{p/2}}{\rho^2}
    \left(
      \log\frac{C_pN}{\rho\eta}
    \right)^{p+5}
    \quad\text{rows}.
\]
\end{theorem}

\begin{proof}
The row bound follows from \cref{prop:one-pass-size} and
\cref{eq:pass-parameter-relations}, since
\[
    L\leq
    C_p\log\!\left(\frac{C_pN}{\rho\eta}\right).
\]
For the approximation guarantee, condition on the current active matrix in
each round.
Because the algorithm continues only when $N_{t-1}>N_{\min}$,
\cref{lem:tuning} applies with $M=N$, target $\theta$, and failure probability
$\eta_0$. Thus, except with conditional probability at most $\eta_0$, the
round satisfies
\[
    C_t(F)=(1\pm\theta)C_{t-1}(F)\pm\theta\OPT_{t-1}
\]
simultaneously for every $F\in\Fset_k$.
Let $T\leq r$ be the number of rounds actually executed.
A union bound over at most $r$ rounds gives all of these events with
probability at least $1-r\eta_0\geq1-\eta/4$.
Since $T\theta\leq r\theta=c\rho$, \cref{lem:error-composition} gives a
$(1\pm\rho)$ multiplicative guarantee when $c$ is sufficiently small.
Combining this event with \cref{prop:one-pass-size} gives success probability
at least $1-3\eta/4\geq1-\eta$.
\end{proof}

\begin{remark}[Implementation and running time]
\label{rem:runtime}
The implementation does not require the exact value of $\lambda_{\A}$ or the
exact ridge leverage scores. Using the routines in~\cite{CMM17,WY25}, one may
compute a constant-factor estimate of $\lambda_{\A}$ together with
high-probability score estimates that, after adjusting constants, satisfy
\[
    \tau_i^{\lambda_{\A}}(\A)
    \leq \widehat\tau_i
    \leq C\tau_i^{\lambda_{\A}}(\A),
    \qquad
    \sum_i\widehat\tau_i=O(k).
\]
Using $\widehat\tau_i$ in the sampling probabilities only increases the
probabilities required by \cref{thm:WY-one-round}, while
\cref{lem:threshold} applies directly to the linear mass
$\sum_i\widehat\tau_i$. Thus every approximation and row-count bound above
remains valid up to constants. The score estimation and sampling primitives
run in $\wtO_p(\nnz(\A)+d^\omega)$ time over the recursive reduction as shown
in~\cite{WY25}. Allocating the unused constant fraction of the failure
budget to the estimation events preserves the success probability in
\cref{thm:one-pass}.
\end{remark}

\subsection{Completion of the \texorpdfstring{$p>2$}{p > 2} construction}

\begin{proof}[Proof of \cref{thm:main-large-p}]
Apply \cref{thm:WY-presparsification} to $\A$ with accuracy $\eps/3$ and
failure probability $\delta/2$. On its success event, the resulting weighted
row subset $\A_0=\Smat_0\A$ satisfies
\begin{equation}
\label{eq:large-p-presparsified-size}
    N_0
    \leq
    C_pk^{p/2}\eps^{-p}
    \left(
      \log\frac{C_pk}{\eps\delta}
    \right)^{C_p}
\end{equation}
and preserves every cost within a factor $1\pm\eps/3$.

Condition on this event and apply \cref{thm:one-pass} once to $\A_0$, with
$\rho=\eps/3$ and $\eta=\delta/2$. The output $\A_1=\Smat_1\A_0$ has at most
\[
    \frac{C_pk^{p/2}}{\eps^2}
    \left(
      \log\frac{C_pN_0}{\eps\delta}
    \right)^{p+5}
\]
rows. By \cref{eq:large-p-presparsified-size}, and because $p$ is fixed,
\[
    \log\frac{C_pN_0}{\eps\delta}
    \leq
    C_p\log\frac{C_pk}{\eps\delta}.
\]
After adjusting $C_p$, this is the row bound in
\cref{thm:main-large-p}.

The two stages fail with total probability at most $\delta$. On their common
success event, \cref{lem:coreset-composition} shows that $\A_1$ preserves all
rank-at-most-$k$ subspace costs within a factor $1\pm\eps$. Moreover,
$\A_1=(\Smat_1\Smat_0)\A$ remains a weighted row subset of the original input.

The first stage takes $\wtO_p(\nnz(\A)+d^\omega)$ time by
\cref{thm:WY-presparsification}. By \cref{rem:runtime}, the second stage takes
$\wtO_p(\nnz(\A_0)+d^\omega)$ time. Since
$\nnz(\A_0)\leq\nnz(\A)$, the total running time is
$\wtO_p(\nnz(\A)+d^\omega)$.
\end{proof}

\section*{Acknowledgements}
The proof for $p>2$ was first obtained using a fully automated Gemini-based
agentic system developed internally at Google. The proof for $1\leq p<2$ was
obtained with contributions from the authors, the aforementioned system, and
a harness built on top of GPT-5.6 Sol. The authors have verified the proofs
and edited them for
clarity of presentation, and take responsibility for the final version. The
authors would like to thank Taisuke Yasuda for helpful discussions.

\bibliographystyle{alpha}
\bibliography{references}

\appendix

\newpage

\section{Auxiliary Process Estimates for
\texorpdfstring{$1\le p<2$}{1 <= p < 2}}
\label{sec:auxiliary-process}

The proof of \cref{thm:split} decomposes each normalized query loss into the
four coordinates in \eqref{eq:four-coordinates}. The directional and
duality-map estimates below control $V_F$ and $H_F$, respectively. The
row-augmented subspace estimate controls $W_F^\perp$, and the final Lewis-weight
estimate controls $T_F$. The scalar composition lemma then recovers the full
loss from these four quantities. We formulate the estimates independently of
the coreset notation so that their hypotheses are transparent when they are
invoked in the main proof.

\subsection{Directional and duality-map classes}

Both classes in this subsection are indexed by linear maps out of an
$r$-dimensional space, but their values may lie in an arbitrarily large
Hilbert space. The estimates therefore retain the parameter $r$ while
avoiding any dependence on the ambient dimension. The first class records the
pairing of the unit direction $J_1(Ay_i)$ with $z_i$. The second replaces
$J_1$ by $J_p$, which also records the magnitude of $Ay_i$. In both proofs,
one Gaussian sketch is shared by the entire function class. After sketching,
the functions are controlled by at most $Lr$ real parameters, independently
of the dimension of the Hilbert space. For the applications in this paper,
the reader may simply take $\cH=\R^D$ for an arbitrary $D$; the Hilbert-space
formulation emphasizes that the bounds do not depend on $D$.

\begin{lemma}[Directional entropy at $p=1$]\label{lem:directional}
Let $r,M\ge1$ and $B\ge0$, let $\cH$ be a real Hilbert space, and let
$\nu$ be a probability law supported on at most $M$ indices. Let
$y_i\in\R^r$ and $z_i\in\cH$ satisfy $\norm{z_i}\le B$. For arbitrary
linear maps $A:\R^r\to\cH$, define
\[
 g_A(i)=\ip{J_1(Ay_i)}{z_i}.
\]
Then, for every $\delta>0$,
\begin{equation}\label{eq:directional-entropy}
 \log N(\delta,\{g_A\},L_2(\nu))
 \le C B^2r\delta^{-2}\log\!\bigl(C(M+r+2)\bigr).
\end{equation}
The convention $J_1(0)=0$ is included.
\end{lemma}

\begin{proof}
Since the zero map $A=0$ gives the zero function and $|g_A(i)|\le B$, a
single $L_2(\nu)$-ball centered at zero covers the class when $\delta\ge B$.
Hence assume $\delta<B$, so
$L=\lceil CB^2\delta^{-2}\rceil\ge2$ after fixing the absolute constant.
Fix an arbitrary finite $\delta$-packing $\cP$.  Only the finitely many
functions in this packing matter, so choose one representing map $A$ for
each of them and use $\cP$ also for this set of representatives. Only the
finitely many vectors $Ay_i$ with $A\in\cP$ and
$i\in\operatorname{supp}(\nu)$, together with the vectors $z_i$ for those
same indices, matter. Let $g_1,\ldots,g_L$ be independent standard
Gaussian vectors in their finite-dimensional span and put
\[
 \widehat g_A(i)=\frac{1}{\gamma_1L}
 \sum_{\ell=1}^L
 \sgn\!\bigl(\ip{g_\ell}{Ay_i}\bigr)\ip{g_\ell}{z_i},
\]
Here $\sgn(0)=0$, and $\gamma_1$ is the Gaussian moment defined in
\eqref{eq:gaussian-projection-moment}.
For $x\ne0$, rotational invariance gives
$\E[\sgn(\ip{g}{x})g]=\gamma_1J_1(x)$; both sides are zero when $x=0$.
Taking the inner product with $z_i$ gives
$\E\widehat g_A(i)=g_A(i)$ for $A\in\cP$, and
$\E|\widehat g_A(i)-g_A(i)|^2\le CB^2/L$. Averaging over $\cP$ and applying
Markov's inequality shows that the expected fraction of maps for which the
$L_2(\nu)$ error exceeds $\delta/8$ is at most $1/8$, after increasing the
constant in $L$. Fix a realization for which at least $7/8$ of the packing is
approximated within $\delta/8$.

For fixed $g_1,\ldots,g_L$, write $X=GA\in\R^{L\times r}$. The coefficients
$\ip{g_\ell}{z_i}$ are now fixed, so the sketch is determined by the ternary
signs of $(Xy_i)_\ell$. As the row $X_{\ell,*}$ varies in $\R^r$, these are
the signs of
$M_0:=|\operatorname{supp}(\nu)|\le M$ linear polynomials in $r$ variables.
Discard the identically zero forms, whose signs are fixed, and let
$M_1\le M_0$ be the number that remain. If $M_1\ge r$, the zero-inclusive
sign-assignment bound of Goldberg and
Jerrum~\cite[Corollary~2.1]{GJ95} gives at most $(8eM_1/r)^r$ patterns; if
$M_1<r$, the trivial bound is $3^{M_1}$. In either case the number is at most
$[C(M+r+2)]^r$. Multiplying over the $L$ rows gives at most
$[C(M+r+2)]^{Lr}$ possible sketches. The retained packing therefore satisfies
the logarithmic bound in \eqref{eq:directional-entropy}, and restoring the
discarded elements changes its size by at most a factor of $8/7$. The same
estimate holds for every finite packing, so a maximal packing yields the
required cover. The sign count includes zero signs, and hence no separate
continuity argument is needed.
\end{proof}

\begin{lemma}[Duality-map entropy]\label{lem:phase}
Let $r,M\ge1$ and $B\ge0$, let $\cH$ be a real Hilbert space, and let
$\nu$ be supported on at most $M$ indices. Let $y_i\in\R^r$,
$\alpha_i\ge0$, and $z_i\in\cH$ satisfy $\norm{z_i}\le B$. Let
$\mathcal A$ be a nonempty family of linear maps $A:\R^r\to\cH$. Fix
$1\le p\le2$ and define
\[
 h_A(i)=p\ip{J_p(\alpha_iAy_i)}{z_i}.
\]
If $1<p\le2$, define the empirical $p$-energy of the class by
\[
 E_\nu=\sup_{A\in\mathcal A}
 \sum_i\nu_i\norm{\alpha_iAy_i}^p,
\]
and assume $E_\nu<\infty$. Then, for every $\delta>0$,
\begin{equation}\label{eq:phase-entropy}
 \log N(\delta,\{h_A\},L_2(\nu))
 \le \widetilde O_{p,B}\!
 \left(rE_\nu^{\,2-2/p}\delta^{-2}\right).
\end{equation}
The logarithms depend only on $M+r+2$, $1+1/\delta$, and $1+E_\nu$. In
this range,
\begin{equation}\label{eq:phase-radius}
 \sup_A\norm{h_A}_{L_2(\nu)}
 \le C_{p,B}E_\nu^{\,1-1/p}.
\end{equation}
At $p=1$, the entropy is
$\widetilde O_B(r\delta^{-2})$ for every $\delta>0$, and the class has
radius $O(B)$.
\end{lemma}

\begin{proof}
Assume first $1<p\le2$, $B>0$, and $E_\nu>0$. Put
$\psi_p(0)=0$ and $\psi_p(t)=|t|^{p-2}t$ for $t\ne0$. Fix an arbitrary finite
$\delta$-packing $\cP$, choose one representing map $A$ for each function,
and use $\cP$ also for this set of representatives. Let $\mathcal H_0$ be
the span of the finitely many vectors $Ay_i$ with $A\in\cP$ and
$i\in\operatorname{supp}(\nu)$, together with the vectors $z_i$ for those
same indices. Let
$g_1,\ldots,g_L$ be independent standard Gaussians in
$\mathcal H_0$ and define
\begin{equation}\label{eq:phase-sketch}
 \widehat h_A(i)=\frac{p}{\gamma_pL}\sum_{\ell=1}^L
 \psi_p\!\bigl(\ip{g_\ell}{\alpha_iAy_i}\bigr)\ip{g_\ell}{z_i}.
\end{equation}
Differentiating
$\E|\ip{g}{x}|^p=\gamma_p\norm{x}^p$ gives
\[
 \E\bigl[\psi_p(\ip{g}{x})g\bigr]=\gamma_pJ_p(x),
\]
and hence $\E\widehat h_A=h_A$. To bound the variance despite the correlation
between the two Gaussian factors, suppose first that $x\ne0$. Write
$\widehat x=x/\norm{x}$ and $z=a\widehat x+z_\perp$, where
$z_\perp\perp x$. For a standard scalar Gaussian $\xi$, independence of the
parallel and orthogonal Gaussian components gives
\[
\begin{aligned}
 \E\left|\psi_p(\ip{g}{x})\ip{g}{z}\right|^2
 &=\norm{x}^{2p-2}
   \left(a^2\E|\xi|^{2p}
       +\norm{z_\perp}^2\E|\xi|^{2p-2}\right)\\
 \le C_p\norm{z}^2\norm{x}^{2p-2}.
\end{aligned}
\]
The same bound is immediate when $x=0$.
Since $\nu$ is a probability law and $2p-2\le p$, monotonicity of moments
gives, for every nonnegative sequence $(a_i)$,
\[
 \sum_i\nu_i a_i^{2p-2}
 \le\left(\sum_i\nu_i a_i^p\right)^{(2p-2)/p}.
\]
Applying this inequality to $a_i=\norm{\alpha_iAy_i}$ yields
\begin{equation}\label{eq:phase-mse}
 \sup_A\E\norm{\widehat h_A-h_A}_{L_2(\nu)}^2
 \le \frac{C_{p,B}E_\nu^{\,2-2/p}}{L}.
\end{equation}
The pointwise bound
$|h_A(i)|\le pB\norm{\alpha_iAy_i}^{p-1}$ and the same moment inequality give
\begin{equation}\label{eq:phase-radius-proof}
 \norm{h_A}_{L_2(\nu)}
 \le pB\left(\sum_i\nu_i\norm{\alpha_iAy_i}^{2p-2}\right)^{1/2}
 \le pB E_\nu^{1-1/p}.
\end{equation}
Thus the class has diameter at most $2pB E_\nu^{1-1/p}$. When
$\delta\ge2pB E_\nu^{1-1/p}$, a single ball centered at any member of the
class gives a cover. In the remaining case, choose
\begin{equation}\label{eq:phase-L}
 L=\left\lceil C_{p,B}E_\nu^{2-2/p}\delta^{-2}\right\rceil.
\end{equation}
The nontrivial-scale assumption makes the quantity inside the ceiling
bounded below by a positive constant depending only on $p,B$.

Write $Gx=(\ip{g_\ell}{x})_{\ell=1}^L$, $X_A=GA$, and equip
$\R^{L\times r}$ with the quotient seminorm
\[
 \norm{X}_{\nu,\alpha,p}
 =\left(\sum_i\nu_i\alpha_i^p\norm{Xy_i}_2^p\right)^{1/p}.
\]
This seminorm can vanish on a nonzero matrix, because only its action on the
vectors $\alpha_i y_i$ in the support of $\nu$ matters. We therefore identify
two matrices when their difference has zero seminorm.
Let $\mathcal N=\{X:\norm{X}_{\nu,\alpha,p}=0\}$.  The quotient
$\mathcal V=\R^{L\times r}/\mathcal N$ is a normed space of dimension at
most $Lr$; the sketched function depends only on the class of $X$ in this
quotient.  Gaussian moments give
\[
 \E\norm{X_A}_{\nu,\alpha,p}^p
 \le C_pL^{p/2}E_\nu.
\]

We next choose one sketch that works for a large subset of the packing.
Increase the constant in \eqref{eq:phase-L} so that, for every $A\in\cP$,
\[
 \Prb\{\norm{\widehat h_A-h_A}_{L_2(\nu)}>\delta/8\}\le1/64,
 \qquad
 \Prb\{\norm{X_A}_{\nu,\alpha,p}>
        C_p\sqrt L E_\nu^{1/p}\}\le1/64.
\]
The first inequality follows from \eqref{eq:phase-mse} and Markov's
inequality, and the second follows from the preceding $p$th-moment estimate.
The expected fraction of packing
elements failing each inequality is at most $1/64$; a second Markov
inequality makes each fraction at most $1/8$ except on an event of
probability at most $1/8$.

For every $i\in\operatorname{supp}(\nu)$, Gaussian concentration for the
norm of a Gaussian vector gives
\[
 \Prb\!\left\{\norm{Gz_i}>
 C B\sqrt{L+\log(8M+8)}\right\}\le\frac1{8M+8}.
\]
A union bound makes this simultaneous envelope event have probability at
least $7/8$. The two fraction bounds and the envelope event therefore hold
simultaneously with positive probability. Fix such a realization of $G$ and
remove the two exceptional subsets. At least $3|\cP|/4$ elements remain;
each retained function is approximated within $\delta/8$, and each retained
$X_A$ lies in the quotient ball
\begin{equation}\label{eq:phase-ball}
 \norm{X_A}_{\nu,\alpha,p}\le C_p\sqrt L\,E_\nu^{1/p}.
\end{equation}
Moreover, the retained sketched functions form a $3\delta/4$-packing by the
triangle inequality.

For a matrix $X$, let $\widehat h_X$ denote the function obtained from
\eqref{eq:phase-sketch} by replacing $X_A$ with $X$. The H\"older estimate
$|\psi_p(a)-\psi_p(b)|\le C_p|a-b|^{p-1}$, followed by Cauchy--Schwarz,
and
$\norm{v}_{2p-2}^{p-1}\le
L^{(2-p)/2}\norm{v}_2^{p-1}$ imply, row by row and then in
$L_2(\nu)$,
\[
\begin{split}
 \norm{\widehat h_X-\widehat h_{X'}}_{L_2(\nu)}
 &\le C_pB\kappa L^{(1-p)/2}
 \left(\sum_i\nu_i
   \norm{\alpha_i(X-X')y_i}^{2p-2}\right)^{1/2}\\
 &\le C_pB\kappa L^{(1-p)/2}
 \norm{X-X'}_{\nu,\alpha,p}^{p-1},
 \qquad
 \kappa=\left(1+\frac{\log(8M+8)}L\right)^{1/2}.
\end{split}
\]
The last step again uses monotonicity of moments under the probability law
$\nu$. Consequently, two quotient points within distance
\[
 \eta=c_{p,B}\sqrt L\,
       (\delta/(B\kappa))^{1/(p-1)}
\]
produce sketched functions within $\delta/8$, after reducing $c_{p,B}$ if
necessary. If $\eta$ exceeds the radius in \eqref{eq:phase-ball}, one net
point suffices. Otherwise, take a maximal $\eta$-separated subset of the ball.
The radius-$\eta/2$ balls around its points are disjoint and lie in the ball
with radius enlarged by $\eta/2$. Comparing their volumes in the at most
$Lr$ dimensional space $\mathcal V$ gives at most
\begin{equation}\label{eq:phase-volume}
 \left(1+
 C_{p,B}E_\nu^{1/p}
       (B\kappa/\delta)^{1/(p-1)}\right)^{Lr}
\end{equation}
mesh points. No two retained $3\delta/4$-separated sketches can be assigned
to the same mesh point. Hence $|\cP|$ is at most twice
\eqref{eq:phase-volume}. Substituting \eqref{eq:phase-L} and absorbing the
logarithm of the parenthesis into
$\polylog(M+r+2,1+\delta^{-1},1+E_\nu)$ proves
\eqref{eq:phase-entropy}.  Since the argument applies to every finite
packing, a maximal packing gives a cover after changing $\delta$ by an
absolute constant.

The asserted radius bound \eqref{eq:phase-radius} follows from
\eqref{eq:phase-radius-proof}. If
$E_\nu=0$, every $\alpha_iAy_i$ vanishes on the support of $\nu$, so the
class is the zero function.  If $B=0$ the same conclusion is immediate.

At $p=1$, positive $\alpha_i$ disappear inside $J_1$, while zero
$\alpha_i$ give a fixed zero coordinate. Apply the packing argument of
\cref{lem:directional} to the vectors $\alpha_i y_i$ (or, equivalently,
delete the zero coordinates). Its use of the same Gaussian sketch for every
function, together with the zero-inclusive sign count,
gives $\widetilde O_B(r\delta^{-2})$, including the strata $Ay_i=0$, and
$|h_A(i)|\le B$ gives the asserted radius.  This also covers $E_\nu=0$ at
the endpoint.
\end{proof}

\begin{corollary}[Conditional duality-map complexity]\label{cor:phase-rad}
For an empirical law $\nu=m^{-1}\sum_{s=1}^m\delta_{I_s}$, the class in
\cref{lem:phase} satisfies, conditionally on the sampled indices,
\[
\Rad_m^{\rm abs}(\{h_A\}\mid I_1,\ldots,I_m)
 \le \widetilde O_{p,B}\!
 \left(\left(E_\nu^{\,1-1/p}
       +E_\nu^{\,(1-1/p)/2}\right)\sqrt{\frac rm}\right)
\]
for $p>1$, and $\widetilde O_B(\sqrt{r/m})$ for $p=1$.
\end{corollary}

\begin{proof}
Suppose first that $p>1$. If $E_\nu=0$, the class is zero by
\cref{lem:phase}. Otherwise apply the truncated Dudley bound
\eqref{eq:standard-dudley} to the class with zero and negatives adjoined,
using \eqref{eq:phase-entropy} and \eqref{eq:phase-radius}. Adjoining these
functions increases each covering number by at most a constant factor.
Truncate the entropy integral at the $L_2(\nu)$ radius divided by $m$ and
write $a=1-1/p$. Up to the stated
logarithmic factors, \eqref{eq:phase-entropy} gives
\[
 \sqrt{\log N(u,\{h_A\},L_2(\nu))}
 \le \widetilde O_{p,B}\!\left(\frac{\sqrt r\,E_\nu^a}{u}\right),
\]
while \eqref{eq:phase-radius} bounds the integration range by
$O_{p,B}(E_\nu^a)$. Substitution into the Dudley integral gives the
$E_\nu^a\sqrt{r/m}$ term. For $E_\nu\le1$, the additional lower-scale
logarithm is absorbed using
$x^a(1+\log(1/x))^c\le C_{a,c}x^{a/2}$; for $E_\nu>1$, it is already among
the logarithms in $1+E_\nu$. At $p=1$, the same truncated Dudley bound
applied to the endpoint entropy and radius in \cref{lem:phase} gives the
stated estimate directly.
\end{proof}

\subsection{Entropy of row-augmented subspace distances}

The next estimate is used for the component of a residual vector orthogonal
to both the query subspace and its associated low-rank row. It bounds the
metric entropy of this class in terms of the query dimension $k$, with no
dependence on the ambient dimension.

Related uses of random projections for affine-space distances, projective
clustering, and subspace approximation appear in
\cite{Mag07,KR15,CW25}. Here we need a different guarantee: an $L_2(q)$
packing bound for the row-dependent augmented spaces
$F+\operatorname{span}\{b_i\}$.

\begin{theorem}[Row-augmented subspace entropy]\label{thm:angular}
Let $\cH$ be a finite-dimensional real Hilbert space and let $q$ be a
finitely supported probability law. Let $u_i\in\cH$ be unit vectors and let
$b_i\in\cH$ be arbitrary fixed vectors, including zero. For every
$k$-dimensional subspace $F\subseteq\cH$, define
\[
 f_F(i)=\dist(u_i,F+\operatorname{span}\{b_i\})^p,\qquad 1\le p\le2.
\]
For $0<\delta\le1$ and $k\ge1$,
\begin{equation}\label{eq:angular-entropy}
 \log N(\delta,\{f_F\},L_2(q))
 \le C_pk\delta^{-2}\polylog(k+2,1/\delta),
\end{equation}
independently of the number of rows and the ambient dimension.
\end{theorem}

\begin{proof}
Fix a finite $\delta$-packing $\cP$. We first compress it to
$k+O_p(\delta^{-2})$ dimensions and retain a constant fraction whose
pairwise distances survive the compression. We then bound the
pseudodimension of the compressed class using Grassmann charts.

Let
\[
 s=\lceil C_p\delta^{-2}\rceil,\qquad L=k+1+s,
\]
and take a Gaussian map $R:\cH\to\R^L$ represented in orthonormal bases by
a matrix with independent $N(0,1/L)$ entries. Put
$G_i(F)=F+\operatorname{span}\{b_i\}$,
$m_i(F)=\dim G_i(F)\in\{k,k+1\}$, and
$c_m=(L/(L-m))^{p/2}$.  Define
\[
 h_{F,R}(i)=c_{m_i(F)}\dist(Ru_i,RG_i(F))^p.
\]
The normalization $c_m$ removes the deterministic shrinkage caused by the
dimension of the orthogonal complement. Without it, the projected distance
contains a factor $(1-m/L)^{p/2}$, whose squared deviation from one need not
be $O(1/s)$ when $k\gg s$.

To identify the distribution of the projected distance, decompose $u_i$ into
its projection onto $G_i(F)$ and an orthogonal residual $v_{i,F}$. Put
$\rho_{i,F}=\norm{v_{i,F}}=\dist(u_i,G_i(F))$ and
$\ell_{i,F}=L-m_i(F)$. Conditional on $R|_{G_i(F)}$, the vector $Rv_{i,F}$
is independent of $RG_i(F)$ and has covariance
$\rho_{i,F}^2I_L/L$. Its projection onto $(RG_i(F))^\perp$ therefore has
$\ell_{i,F}$ Gaussian degrees of freedom, and
\begin{equation}\label{eq:chi-square}
 h_{F,R}(i)=\rho_{i,F}^p
 \left(\frac{X_{\ell_{i,F}}}{\ell_{i,F}}\right)^{p/2},
 \qquad X_{\ell_{i,F}}\sim\chi^2_{\ell_{i,F}},\quad
 \ell_{i,F}\ge s.
\end{equation}
almost surely. Since $\rho_{i,F}\le1$,
$\E(X_\ell/\ell-1)^2=2/\ell$, and
$|x^{p/2}-1|\le|x-1|$ for $0<p/2\le1$, we obtain
\[
 \E_R\norm{h_{F,R}-f_F}_{L_2(q)}^2\le2/s.
\]
Increase the constant in $s$ until the right-hand side is at most
$\delta^2/2^{13}$.  For each $F\in\cP$, Markov gives
\[
 \Prb\{\norm{h_{F,R}-f_F}_{L_2(q)}>\delta/8\}\le1/128.
\]
Thus the expected fraction of elements of $\cP$ that fail this bound is at
most $1/128$, and another application of Markov's inequality makes that
fraction at most $1/8$ except
on an event of probability at most $1/16$.

Since $u_i$ is a unit vector and Gaussian moments are bounded for fixed
$p\le2$, $\E_R\sum_iq_i\norm{Ru_i}^{2p}\le C_p$.  Markov therefore shows
that the event
\begin{equation}\label{eq:angular-envelope-event}
 \sum_iq_i\norm{Ru_i}^{2p}\le C_p.
\end{equation}
has probability at least $7/8$, after enlarging $C_p$. A Gaussian map is
also injective on each of the finitely many relevant spaces $G_i(F)$ with
probability one, since each has dimension at most $k+1\le L$. We may
therefore fix a realization $R$ for which the packing-retention and envelope
events hold and all relevant restrictions are injective. In particular,
$RG_i(F)$ has dimension $k+1$ whenever
$b_i\notin F$. Discard the packing elements that fail the approximation
bound. The retained set $\cP_R$ has size at least $7|\cP|/8$, and its sketched
functions are a $3\delta/4$-packing by the triangle inequality.

Fix this $R$ and put $x_i=Ru_i$ and $a_i=Rb_i$. Injectivity on $G_i(F)$
implies $m_i(RF)=m_i(F)$, so the retained sketches belong to the class
obtained by allowing every $E\in\Gr(k,L)$. Define
\[
 h_E(i)=c_{m_i(E)}
 \dist(x_i,E+\operatorname{span}\{a_i\})^p,
\quad
 m_i(E)=\begin{cases}k,&a_i\in E,\\k+1,&a_i\notin E.\end{cases}
\]
With $c_*=c_{k+1}$ and $R_i^{\rm env}=c_*\norm{x_i}^p$, we have
$0\le h_E(i)\le R_i^{\rm env}$ and
$\mathcal R^2:=\sum_iq_i(R_i^{\rm env})^2\le C_pc_*^2$ by
\eqref{eq:angular-envelope-event}. Set $\eta=\delta/16$. If
$\mathcal R\le\eta$, one ball centered at zero covers the class, so the
retained packing has at most one element. Otherwise, discard the zero-envelope
coordinates and set
\[
 g_E(i)=h_E(i)/R_i^{\rm env},\qquad
 \widehat q_i=q_i(R_i^{\rm env})^2/\mathcal R^2.
\]
Then $0\le g_E\le1$ and
\begin{equation}\label{eq:angular-normalization}
 \norm{h_E-h_{E'}}_{L_2(q)}
 =\mathcal R\norm{g_E-g_{E'}}_{L_2(\widehat q)}.
\end{equation}

We next bound the pseudodimension of the compressed class, treating the
degenerate branch $a_i\in E$ separately.
Write $\ell=L-k=s+1$.  On a standard Grassmann chart parameterize
$E^\perp$ by the columns of
\[
 Z_B=\binom{B}{I_\ell},\qquad B\in\R^{k\times\ell},
\quad
 Q_B=Z_B(Z_B^\top Z_B)^{-1}Z_B^\top.
\]
Here $Q_B$ is the orthogonal projector onto $E^\perp$.
There are $\binom{L}{\ell}$ such charts and
$d_{\rm ch}=k\ell$ parameters per chart.
Let $\Delta=\det(I_\ell+B^\top B)>0$.  For fixed $x=x_i$, $a=a_i$, write
\[
 x^\top Q_Bx=N_X/\Delta,\quad
 a^\top Q_Ba=N_A/\Delta,\quad
 x^\top Q_Ba=N_C/\Delta,
\]
where the numerators have degree $O(\ell)$ and $N_A\ge0$.  If $N_A>0$,
\begin{equation}\label{eq:angular-ratio}
 \dist(x,E+\operatorname{span}\{a\})^2
 =\frac{N_XN_A-N_C^2}{\Delta N_A};
\end{equation}
if $N_A=0$, equivalently $a\in E$, the squared distance is $N_X/\Delta$.
For a pseudodimension threshold $\tau_i$ at sample point $i$, the inequality
$g_E(i)>\tau_i$ is identically true when $\tau_i<0$ and identically false
when $\tau_i\ge1$.
For $0\le\tau_i<1$, set, separately for $m=k,k+1$,
\[
 \theta_{i,m}=\left(\frac{\tau_iR_i^{\rm env}}{c_m}\right)^{2/p}.
\]
After raising the threshold inequality to the power $2/p$,
$g_E(i)>\tau_i$ is equivalent to
\[
 \bigl(N_A=0\ \text{and}\
       N_X-\theta_{i,k}\Delta>0\bigr)
 \quad\text{or}\quad
 \bigl(N_A>0\ \text{and}\
       N_XN_A-N_C^2-\theta_{i,k+1}\Delta N_A>0\bigr).
\]
Thus the zero-denominator branch is included as a separate semialgebraic
stratum. The thresholds may vary with $i$, as required in the definition of
pseudodimension.
If $v\le d_{\rm ch}$, the desired pseudodimension bound is immediate, so
assume $v>d_{\rm ch}$. Across $v$ thresholded rows, the predicates on one
chart involve $O(v)$ polynomials of degree $O(\ell)$ in $d_{\rm ch}$ real
parameters. The zero-inclusive sign-assignment bound of Goldberg and
Jerrum~\cite[Corollary~2.1]{GJ95} gives at most
\[
 \left(\frac{C\ell v}{d_{\rm ch}}\right)^{d_{\rm ch}}
 \le (C\ell v)^{d_{\rm ch}}
\]
labelings on one chart. Summing over the charts gives at most
\[
 \binom{L}{\ell}(C\ell v)^{d_{\rm ch}}
\]
labelings.  If $v$ points were shattered, then
\[
 v\le\log_2\binom L\ell+d_{\rm ch}\log_2(C\ell v).
\]
Since
$\log\binom L\ell\le\ell\log(eL/\ell)
\le C d_{\rm ch}\log(L+2)$ and $d_{\rm ch}=k\ell\le L^2$, solving this
inequality for $v$ gives
\[
 \Pdim(\{g_E\})\le C d_{\rm ch}\log(L+2).
\]
The pseudodimension entropy bound \eqref{eq:pdim-entropy} and
\eqref{eq:angular-normalization}, with
$c_*=[1+(k+1)/s]^{p/2}$, give
\[
 \log N(\eta,\{h_E\},L_2(q))
 \le Ck(s+1)\log(L+2)\log(C\mathcal R/\eta).
\]
Indeed $\mathcal R\le C_pc_*$ by \eqref{eq:angular-envelope-event}, and
\[
 \log(C\mathcal R/\eta)
 \le C_p\bigl(\log(k+2)+\log(1/\delta)\bigr).
\]
Since $s=O_p(\delta^{-2})$, this is
the right-hand side of \eqref{eq:angular-entropy}. Each $L_2(q)$-ball of
radius $\eta$ contains at most one member of the retained
$3\delta/4$-packing. Hence the cover bounds $|\cP_R|$, and
$|\cP|\le(8/7)|\cP_R|$ gives the same logarithmic estimate for the original
packing, up to an additive constant. The estimate holds for every finite
$\delta$-packing, and a maximal packing gives the required $\delta$-cover.
\end{proof}

\begin{corollary}[Rank-at-most row-augmented entropy]\label{cor:angular}
Let $\cH$ be a finite-dimensional real Hilbert space, let $q$ be a finitely
supported probability law, let $u_i\in\cH$ be unit vectors, and let
$b_i\in\cH$ be fixed vectors. Fix $k\ge1$ and $B,T_0>0$.
Suppose that $0\le\beta_i^p\le B$, and associate with every subspace
$F\subseteq\cH$ of dimension at most $k$ a scalar $t_F\in[0,T_0]$. Then the
class
\[
 i\longmapsto t_F^p\beta_i^p
 \dist(u_i,F+\operatorname{span}\{b_i\})^p
 \qquad(F\subseteq\cH,\ \dim F\le k)
\]
has
\[
 \log N(\delta,\mathcal G,L_2(q))
 \le \widetilde O_{p,B,T_0}(k\delta^{-2})
 \qquad(0<\delta\le1),
\]
where $\mathcal G$ denotes the displayed class.
\end{corollary}

\begin{proof}
For each $1\le j\le k$, apply \cref{thm:angular} to the stratum
$\dim F=j$. Indeed, for $a,a'\in[0,T_0^p]$ and two such subspaces $F,F'$,
the triangle inequality gives
\[
 \norm{a\beta^pf_F-a'\beta^pf_{F'}}_{L_2(q)}
 \le BT_0^p\norm{f_F-f_{F'}}_{L_2(q)}+B|a-a'|.
\]
Here multiplication by $\beta^p$ is row-wise.
Thus a suitably rescaled cover from \cref{thm:angular}, together with a
grid of mesh $c_{B,T_0}\delta$ for $a=t_F^p$, covers this stratum. The
$j=0$ stratum contains only $F=\{0\}$. Taking the union of the $k+1$ covers
adds only $\log(k+1)$ to the entropy.
\end{proof}

\subsection{A scalar composition lemma}

The following deterministic lemma shows that the parallel power is a
Lipschitz function of the first three scalar coordinates introduced in
\cref{sec:four-coordinate-decomposition}. Its bound includes zero amplitudes,
both signs of the parallel residual, and exact cancellation.

\begin{lemma}[Lipschitz scalar decoder]\label{lem:composition}
Fix $1\le p\le2$.  For $a,b\ge0$ and $\sigma\in\{-1,1\}$ define
\[
 T=a^p,\qquad V=b^p,\qquad H=p\sigma a^{p-1}b,\qquad
 U=|a+\sigma b|^p.
\]
At $p=1$ set $H=0$ when $a=0$. For another
$a',b'\ge0$ and $\sigma'\in\{-1,1\}$, define $T',V',H',U'$ analogously.
Then
\begin{equation}\label{eq:composition}
 |U-U'|\le C_p\bigl(|T-T'|+|V-V'|+|H-H'|\bigr),
\end{equation}
where $C_p$ depends only on $p$. This includes sign changes, exact
cancellation, and zero amplitudes.
\end{lemma}

\begin{proof}
Assume first $1<p\le2$.  Put $S=T+V$.  When $S>0$, let
\[
 t=T/S,\qquad
 \alpha(t)=t^{1/p},\qquad \beta(t)=(1-t)^{1/p},
\]
\[
 h(t)=p\,t^{(p-1)/p}(1-t)^{1/p},\qquad
 u_\sigma(t)=|\alpha(t)+\sigma\beta(t)|^p.
\]
Then $(T,V,H,U)=S(t,1-t,\sigma h(t),u_\sigma(t))$.  We claim
\begin{equation}\label{eq:normalized-curve}
 |u_\sigma(t)-u_\sigma(s)|
 \le C_p\bigl(|t-s|+|h(t)-h(s)|\bigr).
\end{equation}
Write $\rho=(p-1)/p$. In $(0,1)$,
\[
 u_+'(t)=(\alpha+\beta)^{p-1}
            [t^{-\rho}-(1-t)^{-\rho}],
\]
\[
 |u_-'(t)|=|\alpha-\beta|^{p-1}
            [t^{-\rho}+(1-t)^{-\rho}],
\qquad
 \frac{h'(t)}{h(t)}
 =\frac{\rho-t}{t(1-t)}.
\]
For fixed $p$, these formulas give
\[
 |u_\sigma'(t)|\le C_p(1+|h'(t)|)
\]
on $(0,\rho/2]$ and $[3/4,1)$: near zero, $t^{-\rho}$ is controlled by
$|h'(t)|$, and near one the same is true of $(1-t)^{-\rho}$. Both
$u_\sigma$ and $h$ are $C_p$-Lipschitz on the middle interval
$[\rho/2,3/4]$, while $h$ is monotone on each endpoint interval. Thus,
integrating on any one of the three intervals gives
\eqref{eq:normalized-curve}. If $[s,t]$ crosses one boundary, apply these
bounds on the two pieces and use the Lipschitz bound for $h$ on the middle
piece. If it crosses both boundaries, then
$|t-s|\ge3/4-\rho/2\ge1/2$, and the claim follows directly from the
boundedness of $u_\sigma$. This proves \eqref{eq:normalized-curve} on
$[0,1]$, including the cancellation point of $u_-$, by continuity.

For two tuples of the same sign, say $S\ge S'>0$,
\[
 |Su_\sigma(t)-S'u_\sigma(t')|
 \le C_p|S-S'|
 +C_pS'\bigl(|t-t'|+|h(t)-h(t')|\bigr).
\]
Moreover,
\[
 |S-S'|\le|T-T'|+|V-V'|,\qquad
 S'|t-t'|\le|T-T'|+|S-S'|,
\]
\[
 S'|h(t)-h(t')|\le|H-H'|+2|S-S'|.
\]
This proves \eqref{eq:composition} for equal signs, including $S'=0$ by a
direct limit.

Now suppose that $\sigma'=-\sigma$. Keep the sign $\sigma$ while changing
the amplitudes, and define
\[
 \widetilde U'=|a'+\sigma b'|^p,
 \qquad
 \widetilde H'=p\sigma(a')^{p-1}b'=-H'.
\]
The equal-sign case gives
\[
 |U-\widetilde U'|
 \le C_p\bigl(|T-T'|+|V-V'|+|H-\widetilde H'|\bigr).
\]
Since
\[
 |H-\widetilde H'|
 =p\bigl|a^{p-1}b-(a')^{p-1}b'\bigr|
 \le |H-H'|,
\]
it remains only to flip the sign of the second tuple. For $c,d\ge0$,
\[
 (c+d)^p-|c-d|^p\le C_p c^{p-1}d.
\]
Indeed, this follows from the mean-value theorem when $d\le c$. When
$d\ge c$, the mean-value theorem gives the bound $C_p d^{p-1}c$, and
$d^{p-1}c\le c^{p-1}d$ because $p\le2$.  Hence in all cases the sign-flip
cost satisfies
\[
 |\widetilde U'-U'|
 \le C_p(a')^{p-1}b'
 \le C_p|H-H'|.
\]
The triangle inequality proves \eqref{eq:composition} for opposite signs.

At $p=1$, if $a,a'>0$, the reverse triangle inequality gives
\[
 \bigl||a+\sigma b|-|a'+\sigma'b'|\bigr|
 \le|a-a'|+|\sigma b-\sigma'b'|
 =|T-T'|+|H-H'|.
\]
If $a=0$, then $H=0$ by convention and $U=b=V$, independently of $\sigma$.
Thus, when $a=0<a'$,
\[
 |U-U'|
 \le |b-b'|+\bigl|b'-|a'+\sigma'b'|\bigr|
 \le |V-V'|+|T-T'|.
\]
If $a=a'=0$, then $|U-U'|=|V-V'|$. The remaining asymmetric case follows
by interchanging the two tuples.
\end{proof}

\subsection{Lewis-weight control of the low-rank term}

This section bounds the empirical complexity of normalized low-rank
costs. We first use the arbitrary-sampling formulation of Cohen and
Peng~\cite[Theorem~7.1]{CP15} and their moment reduction
\cite[Lemma~7.4]{CP15} to extract a scalar first-moment estimate from the
Lewis-weight concentration bounds of Talagrand~\cite{Tal90,Tal95}. We then
derive a scalar Rademacher bound and lift it to Hilbert-valued linear maps by
a Gaussian-mixture representation. The resulting estimate has no dependence
on the dimension of the output space.

\begin{corollary}[First-moment Cohen--Peng bound]
\label{cor:CP-first-moment}
Let $1\le p<2$ and $r\ge1$, and suppose $(w_i,y_i)_{i\in[n]}$ with
$w_i\ge0$ satisfy
$\sum_iw_iy_iy_i^\top=I_r$ with unit $y_i$ on the positive support. Let
$q$ be a probability distribution with $q_i>0$ whenever $w_i>0$. Put
\[
 \mathcal L(x)=\sum_iw_i|\ip{y_i}{x}|^p,
 \qquad
 \ell_x(i)=\frac{w_i}{q_i}|\ip{y_i}{x}|^p,
 \qquad
 K=\max_{w_i>0}\frac{w_i}{q_i},
\]
with $\ell_x(i)=0$ when $w_i=q_i=0$. For i.i.d. indices
$I_1,\ldots,I_j\sim q$, put
\[
 P_jf=\frac1j\sum_{s=1}^jf(I_s),
 \qquad
 \Delta_j=\E_I\sup_{\mathcal L(x)=1}|P_j\ell_x-1|.
\]
There are constants $C_p,c_p$ such that, for $0<\eta\le1/2$ and
$N_0=n+r+j+K+\eta^{-1}+2$,
\begin{equation}\label{eq:CP-first-moment}
 \frac jK\ge C_p\eta^{-2}\log^{c_p}N_0
 \quad\Longrightarrow\quad
 \Delta_j\le C_p\eta.
\end{equation}
\end{corollary}

\begin{proof}
Let $\mathcal A$ have rows $w_i^{1/p}y_i^\top$. After deleting its zero
rows, the Gram matrix in the defining Lewis-weight equation is
\[
 \sum_{i:w_i>0}w_i^{2/p}w_i^{1-2/p}y_iy_i^\top
 =\sum_iw_iy_iy_i^\top=I_r,
\]
and the corresponding quadratic form for row $i$ is
$w_i^{2/p}\norm{y_i}_2^2=w_i^{2/p}$. Thus the $\ell_p$ Lewis weights of
$\mathcal A$ are $w_i$. Moreover,
$\mathcal L(x)=\norm{\mathcal A x}_p^p$.

In the sampling notation of Cohen and
Peng~\cite[Theorem~7.1]{CP15}, set $\pi_i=jq_i$. Since
$\sum_i\pi_i=j$, their procedure draws $j$ rows independently, chooses row
$i$ with probability $q_i$, and scales it by $\pi_i^{-1/p}$. Hence the
sampled $p$th-power norm is
\[
 \norm{S\mathcal A x}_p^p
 =\frac1j\sum_{s=1}^j
   \frac{w_{I_s}}{q_{I_s}}|\ip{y_{I_s}}{x}|^p
 =P_j\ell_x,
 \qquad
 \min_{w_i>0}\frac{\pi_i}{w_i}=\frac jK.
\]
Apply the moment reduction \cite[Lemma~7.4]{CP15} with $\ell=1$ and a fixed
constant value of its parameter $\delta$. The required first-moment
Rademacher bounds are supplied by
\cite[Lemma~7.2]{CP15} at $p=1$ and \cite[Lemma~7.3]{CP15} for $1<p<2$.
These bounds require a target Lewis-weight scale
$\bar u=c_p\eta^2/\log^{c_p}N_0$.

The statement of \cite[Lemma~7.4]{CP15} is presented under the temporary
restriction $\bar u=\Omega(1/r)$. Its proof uses the auxiliary matrix from
\cite[Lemma~B.1]{CP15}, which has $O(r^2)$ rows with Lewis weights $O(1/r)$.
When $\bar u<c/r$, split each auxiliary row into
$t=\lceil C/(r\bar u)\rceil$ collinear copies, each scaled by $t^{-1/p}$.
This operation preserves its contribution to the $p$th-power norm and divides
its Lewis weight by $t$, so every auxiliary copy has weight at most
$\bar u$. It also leaves the Lewis Gram matrix unchanged: if the original row
is $a$ with Lewis weight $w$, then the $t$ copies contribute
\[
 t\left(\frac wt\right)^{1-2/p}
   (t^{-1/p}a)(t^{-1/p}a)^\top
 =w^{1-2/p}aa^\top.
\]
Appending the sampled rows cannot increase these
weights~\cite[Lemma~5.5]{CP15}.

Each occurrence of sampled row $i$ carries comparison weight $w_i/\pi_i$.
The hypothesis in \eqref{eq:CP-first-moment} gives
\[
 \frac{w_i}{\pi_i}\le\frac Kj\le\bar u
\]
after adjusting the constants and logarithmic power. This verifies the
required sampling condition in \cite[Lemma~7.4]{CP15}. The resulting
augmented matrix has
\[
 j+O(r^2t)
 =O\left(j+r^2+\frac r{\bar u}\right)
 \le N_0^{C_p}.
\]
rows, so its logarithm remains $O_p(\log N_0)$. The auxiliary rows are used
only in the comparison argument and do not change the sampling law $q$.
Consequently,
\cite[Lemmas~7.2--7.4]{CP15} give
$\Delta_j\le C_p\eta$ under the condition in
\eqref{eq:CP-first-moment}.
\end{proof}

\begin{lemma}[Scalar Lewis-weight empirical bound]\label{lem:scalar-empirical}
Let $1\le p<2$ and $r\ge1$, let $(w_i,y_i)_{i\in[n]}$ with $w_i\ge0$
satisfy
$\sum_iw_iy_iy_i^\top=I_r$ with unit $y_i$ on the positive support, and let
$q$ be a probability distribution with $q_i>0$ whenever $w_i>0$. Put
\[
 \mathcal L(x)=\sum_iw_i|\ip{y_i}{x}|^p,
 \qquad
 \ell_x(i)=\frac{w_i}{q_i}|\ip{y_i}{x}|^p
 \quad(\mathcal L(x)=1).
\]
Again $\ell_x(i)=0$ when $w_i=q_i=0$. If
$K=\max_{w_i>0}w_i/q_i$, then
\begin{equation}\label{eq:scalar-empirical}
 \E_{I,\eps}\sup_{\mathcal L(x)=1}
 \left|\frac1m\sum_{s=1}^m\eps_s\ell_x(I_s)\right|
 \le C_p\Lambda_p(n+r+m+K+2)
 \left(\sqrt{\frac Km}+\frac Km\right),
\end{equation}
where $I_s$ are i.i.d. with law $q$ and
$\Lambda_p(N)=\log^{c_p}N$ for a constant $c_p$ depending only on $p$.
The absolute value is inside the supremum, and the sampling is with
replacement.
\end{lemma}

\begin{proof}
Let $\lambda_p(N)=\log^{d_p}N$, where $d_p$ is large enough to absorb the
logarithmic powers in \cref{cor:CP-first-moment}, and put
\[
 \Delta_j=\E_I\sup_{\mathcal L(x)=1}|P_j\ell_x-1|.
\]
We first record a pointwise envelope. For $x\ne0$, the bound
$|\ip{y_i}{x}|\le\norm{x}$ and the inequality $p-2\le0$ give
\[
 \mathcal L(x)=\sum_iw_i|\ip{y_i}{x}|^p
 \ge\norm{x}^{p-2}\sum_iw_i|\ip{y_i}{x}|^2
 =\norm{x}^p.
\]
Since $\norm{y_i}=1$ on the positive support,
$0\le\ell_x(i)\le K$. Also,
\[
 r=\sum_{i:w_i>0}q_i\frac{w_i}{q_i}
 \le K\sum_{i:w_i>0}q_i\le K,
\]
so $K\ge r\ge1$.
The first-moment estimate in \cref{cor:CP-first-moment} implies that
\begin{equation}\label{eq:Delta-large}
 \Delta_j\le C_p\lambda_p(n+r+j+K+2)\sqrt{K/j}
\end{equation}
whenever $j\ge C_pK\lambda_p(n+r+j+K+2)^2$. Indeed, set
$\eta=C_p\lambda_p(n+r+j+K+2)\sqrt{K/j}$. A sufficiently large constant in
the lower bound on $j$ ensures that $\eta\le1/2$, and the term $\eta^{-1}$
in \eqref{eq:CP-first-moment} changes only the logarithmic factor.

We now pass from the unsigned empirical error to the signed process. Let
$M_+=|\{s:\eps_s=1\}|$ and $M_-=m-M_+$. Conditional on the signs,
the two groups consist of independent samples from $q$. Write
$P_{M_+}^+$ and $P_{M_-}^-$ for their empirical averages. When both groups
are nonempty,
\begin{align*}
 \frac1m\sum_{s=1}^m\eps_s\ell_x(I_s)
  ={}&\frac{M_+}{m}(P_{M_+}^+\ell_x-1)
     -\frac{M_-}{m}(P_{M_-}^-\ell_x-1)
     +\frac{M_+-M_-}{m}.
\end{align*}
Put
\[
 L=\lambda_p(n+r+m+K+2),
 \qquad
 j_0=\left\lceil C_pKL^2\right\rceil,
\]
with a sufficiently large constant $C_p$. If $m\ge4j_0$, then
on the event $\{M_+,M_-\ge m/4\}$, \eqref{eq:Delta-large} applies to both
groups. The complementary event has probability at most $2e^{-m/16}$, the
supremum is at most $K$, and its contribution is absorbed because
$Ke^{-m/16}\le C_p\sqrt{K/m}$ in this range. Moreover,
$\E|M_+-M_-|/m\le m^{-1/2}\le\sqrt{K/m}$. Taking expectations in the
preceding decomposition therefore gives
\[
 \E_{I,\eps}\sup_{\mathcal L(x)=1}
 \left|\frac1m\sum_{s=1}^m\eps_s\ell_x(I_s)\right|
 \le C_pL\sqrt{\frac Km}.
\]

It remains to treat $m<4j_0$. Write
\[
 a_j=\E\sup_{\mathcal L(x)=1}
 \left|\sum_{s=1}^j\eps_s\ell_x(I_s)\right|.
\]
The map $f\mapsto\sup_{\mathcal L(x)=1}|f(x)|$ is convex, and the next signed
summand has conditional mean zero. Conditional Jensen therefore gives
$a_j\le a_{j+1}$. Set $m_0=4j_0$.
Since $m_0$ is at most a fixed polylogarithmic multiple of
$n+r+m+K+2$,
\[
 \lambda_p(n+r+m_0+K+2)\le C_pL.
\]
Choosing the constant in $j_0$ sufficiently large makes $m_0$ satisfy the
large-sample condition with its own logarithmic factor. The bound just proved
then gives
\[
 a_{m_0}
 \le C_p\lambda_p(n+r+m_0+K+2)\sqrt{Km_0}
 \le C_pKL^2.
\]
By monotonicity,
\[
 \frac{a_m}{m}\le\frac{a_{m_0}}m
 \le C_pL^2\frac Km.
\]
Choose the exponent $c_p$ in the statement so that
$\Lambda_p(N)=\log^{c_p}N$ dominates $\lambda_p(N)^2$. The large- and
small-sample bounds then give \eqref{eq:scalar-empirical}.
\end{proof}

\begin{theorem}[Lewis-weight bound for low-rank terms]
\label{thm:low-rank-lewis}
Let $1\le p<2$ and $r\ge1$. For $i\in[n]$, let $w_i\ge0$ and
$y_i\in\R^r$ satisfy
\[
 \sum_{i=1}^n w_i y_i y_i^\top=I_r,
 \qquad
 \norm{y_i}_2=1\quad\text{whenever }w_i>0,
\]
and set $y_i=0$ whenever $w_i=0$. Let $q$ be any probability distribution
with $q_i>0$ whenever $w_i>0$, and put
$K=\max_{w_i>0}w_i/q_i$.  For a linear map $T:\R^r\to\cH$ with
$S_T=\sum_jw_j\norm{Ty_j}^p>0$, define
\[
 h_T(i)=
 \begin{cases}
  \dfrac{w_i\norm{Ty_i}^p}{q_iS_T},&w_i>0,\\[4pt]
  0,&w_i=0.
 \end{cases}
\]
Then
\begin{equation}
 \Rad_m^{\rm abs}(\{h_T\})
 \le C_p\polylog(n+r+m+K+2)
 \left(\sqrt{\frac Km}+\frac Km\right).
 \label{eq:low-rank-rad}
\end{equation}
Zero-cost maps are assigned the zero function.
\end{theorem}

\begin{proof}
We use a Gaussian scalarization of Euclidean norms. A closely related
reduction from matrix-valued $(p,2)$-norm sampling to scalar $\ell_p$
sampling appears in~\cite[Lemma~3.2]{WY25}. Here the same idea represents
each normalized Hilbert-valued function $h_T$ as a mixture of scalar
functions to which \cref{lem:scalar-empirical} applies. Put
\[
 \mathcal L(x)=\sum_iw_i\abs{\ip{y_i}{x}}^p.
\]
For $\mathcal L(x)=1$ put
\[
 \ell_x(i)=
 \begin{cases}
  \dfrac{w_i}{q_i}\abs{\ip{y_i}{x}}^p,&w_i>0,\\[4pt]
  0,&w_i=0.
 \end{cases}
\]
Let $\gamma_T$ be standard Gaussian probability measure on
$\operatorname{span}\{Ty_i\}$, let $g\sim\gamma_T$, and set $x_g=T^*g$.
By \eqref{eq:gaussian-projection-moment},
$\E_g\mathcal L(x_g)=\gamma_pS_T$. Define the probability measure
\[
 d\mu_T(g)=\frac{\mathcal L(x_g)}{\gamma_pS_T}\,d\gamma_T(g),
 \qquad
 \widehat x_g=\frac{x_g}{\mathcal L(x_g)^{1/p}}
 \quad\text{when }\mathcal L(x_g)>0.
\]
Define $\widehat x_g$ arbitrarily on the remaining set, which has zero
$\mu_T$-measure. For $w_i>0$,
\begin{equation}\label{eq:tilted-mixture}
\begin{aligned}
 \int\ell_{\widehat x_g}(i)\,d\mu_T(g)
 &=\frac{w_i}{q_i\gamma_pS_T}
   \E_g|\ip{Ty_i}{g}|^p\\
 &=\frac{w_i\norm{Ty_i}^p}{q_iS_T}
 =h_T(i).
\end{aligned}
\end{equation}
When $w_i=0$, both sides of \eqref{eq:tilted-mixture} are zero.
Equation~\eqref{eq:tilted-mixture} places every $h_T$ in the closed convex
hull of $\{\ell_x:\mathcal L(x)=1\}$. For each fixed sample and sign vector,
the map
$f\mapsto|m^{-1}\sum_s\eps_sf(I_s)|$ is convex, so its supremum over
this convex hull is no larger than its supremum over the scalar class.
Therefore,
\[
 \sup_T\left|\frac1m\sum_s\eps_sh_T(I_s)\right|
 \le\sup_{\mathcal L(x)=1}\left|\frac1m\sum_s\eps_s\ell_x(I_s)\right|,
\]
and \cref{lem:scalar-empirical} proves \eqref{eq:low-rank-rad}.
\end{proof}

\begin{remark}[Necessity of the $K/m$ term]
The linear term in \eqref{eq:low-rank-rad} cannot in general be omitted.
Take $n=r$, $y_i=e_i$, $w_i=1$, and $q_i=1/r$ for $i\in[r]$. Then $K=r$
and the class contains $r\1_{\{i=j\}}$ for every $j\in[r]$, so its
one-sample absolute complexity is $r$.
\end{remark}

\end{document}